\documentclass[smallextended]{svjour3}  
\smartqed  
\usepackage{latexsym}

\usepackage{graphicx}
\usepackage{caption}
\usepackage{subcaption}
\usepackage{siunitx}
\usepackage{booktabs}
\usepackage[numbers]{natbib}

\usepackage[misc,geometry]{ifsym} 
\usepackage{scrtime}
\usepackage[dont-mess-around]{fnpct}
\usepackage{etoolbox}
\usepackage{amsmath}
\usepackage{xcolor}
\usepackage[linesnumbered,ruled,vlined]{algorithm2e}
\usepackage{docmute}
\usepackage{import}
\usepackage{mathtools}
\usepackage{multirow}
\usepackage[utf8]{inputenc}
\usepackage[english]{babel}
\usepackage{hyperref}
\usepackage{cleveref}
\usepackage{pbox}
\usepackage[shortlabels]{enumitem}

\newcommand*{\affaddr}[1]{#1} 
\newcommand*{\affmark}[1][*]{\textsuperscript{#1}} 
\newcommand{\RNum}[1]{\uppercase\expandafter{\romannumeral #1\relax}}

\SetDataSty{myargfont}
\SetCommentSty{mycommfont}
\usepackage[labelformat=simple]{subcaption}

\Crefformat{figure}{#2Fig.~#1#3}
\Crefmultiformat{figure}{Figs.~#2#1#3}{ and~#2#1#3}{, #2#1#3}{ and~#2#1#3}
\crefrangeformat{figure}{figs. #3#1#4 to #5#2#6}
\Crefformat{table}{#2Table~#1#3}
\Crefmultiformat{table}{Table~#2#1#3}{ and~#2#1#3}{, #2#1#3}{ and~#2#1#3}
\crefrangeformat{table}{table #3#1#4 to #5#2#6}
\Crefformat{equation}{#2Eq.~(#1#3)}
\crefrangeformat{equation}{eqs. #3(#1)#4 to #5(#2)#6}
\Crefmultiformat{equation}{Eqs.~#2(#1)#3}{ and~#2(#1)#3}{, #2(#1)#3}{ and~#2(#1)#3}

\Crefformat{algorithm}{#2Algorithm~#1#3}
\Crefmultiformat{algorithm}{Algorithms~#2#1#3}{ and~#2#1#3}{, #2#1#3}{ and~#2#1#3}
\crefrangeformat{algorithm}{Algorithms #3(#1)#4 to #5(#2)#6}

\Crefformat{lemma}{#2Lemma~#1#3}
\Crefmultiformat{lemma}{Lemmas~#2#1#3}{ and~#2#1#3}{, #2#1#3}{ and~#2#1#3}
\crefrangeformat{lemma}{Lemmas #3#1#4 to #5#2#6}

\Crefformat{section}{#2Section~#1#3}
\Crefmultiformat{section}{Sections~#2#1#3}{ and~#2#1#3}{, #2#1#3}{ and~#2#1#3}
\crefrangeformat{section}{Sections #3#1#4 to #5#2#6}
\Crefformat{subsection}{#2Subsection~#1#3}
\Crefmultiformat{subsection}{Subsections~#2#1#3}{ and~#2#1#3}{, #2#1#3}{ and~#2#1#3}
\crefrangeformat{subsection}{Subsections #3#1#4 to #5#2#6}

\SetKwInput{Input}{Input}
\SetKwInOut{Output}{Output}

\SetKwData{VmSet}{$vmSet$}
\SetKwData{ServerSet}{$serverSet$}
\SetKwData{C}{CS}
\SetKwData{CSet}{SCS}
\SetKwData{Cs}{CSs}
\SetKwData{L}{LS}
\SetKwData{c}{$cs$}
\SetKwData{cset}{$scs$}
\SetKwData{l}{$ls$}
\SetKwData{LA}{$\leftarrow$}

\SetKwFunction{setStartTime}{setStartTime}
\SetKwFunction{setEndTime}{setEndTime}
\SetKwFunction{calculateValue}{calculateValue}
\SetKwFunction{findVMWithEarliestStartTime}{findVMWithEarliestStartTime}
\SetKwFunction{findTimeMostVmsContaining}{findTimeMostVmsContaining}
\SetKwFunction{getVMSetContainingTime}{getVMSetContainingTime}
\SetKwFunction{descendingSortByDuration}{descendingSortByDuration}
\SetKwFunction{shuffle}{shuffle}
\SetKwFunction{canAccommodate}{canAccommodate}
\SetKwFunction{accommodate}{accommodate}
\SetKwFunction{BB}{BB}
\SetKwFunction{DDFF}{DDFF}
\SetKwFunction{p}{p}
\SetKwFunction{main}{main}
\SetKwFunction{SingleSliding}{SingleSliding}
\SetKwFunction{SplitVirtualMachineRequest}{SplitVirtualMachineRequest}
\SetKwFunction{addLeftChild}{addLeftChild}
\SetKwFunction{addRightChild}{addRightChild}
\SetKwFunction{add}{add}

\SetKw{cluster}{cluster}
\SetKw{T}{T}   

\begin{document}
	
	\title{Reducing the Upfront Cost of Private Clouds with Clairvoyant Virtual Machine Placement
		\thanks{The work described in this paper was supported by the National High-tech R\&D Program of China (863 Program) under Grant 2013AA01A215 and the National Laboratory of High-effect Server and Storage Techniques under Grant 2014HSSA05.}}
	
	\titlerunning{Reducing the Upfront Cost of Private Clouds with Clairvoyant VMP}        
	
	\author{Yan Zhao \protect\affmark[1]        \and
		Hongwei Liu  \Letter \affmark[1]     \and
		Yan Wang \affmark[2]      \and
		Zhan Zhang \affmark[1]			 \and
		Decheng Zuo \affmark[1]
	}
	
	
	\institute{Yan Zhao \\
		\email{yanzhao@hit.edu.cn}           
		\and
		\\Hongwei Liu  \\
		\email{liuhw@hit.edu.cn} 
		\and
		\\Yan Wang  \\
		\email{yan.wang@mq.edu.au} 
		\and
		\\Zhan Zhang  \\
		\email{zz@ftcl.hit.edu.cn} 
		\and
		\\Decheng Zuo \\
		\email{zuodc@hit.edu.cn}      
		\\ \affaddr{\affmark[1] Department of Computer Science and Technology, Harbin Institute of Technology, Heilongjiang, China\\
		\affmark[2] Department of Computing, Macquarie University, Sydney, Australia}\\
	}
	
	\date{Received: date / Accepted: date}

	\maketitle
	
	\begin{abstract}
		Although public clouds still occupy the largest portion of the total cloud infrastructure, private clouds are attracting increasing interest from both industry and academia because of their better security and privacy control. According to the existing studies, the high upfront cost is among the most critical challenges associated with private clouds. To reduce cost and improve performance, virtual machine placement (VMP) methods have been extensively investigated; however, few of these methods have focused on private clouds. This paper proposes a heterogeneous and multidimensional clairvoyant dynamic bin packing (CDBP) model, in which the scheduler can conduct more efficient VMP processes using additional information on the arrival time and duration of virtual machines to reduce the datacenter scale and thereby decrease the upfront cost of private clouds. In addition, a novel branch-and-bound algorithm with a divide-and-conquer strategy (DCBB) is proposed to effectively and efficiently handle the derived problem. One state-of-the-art and several classic VMP methods are also modified to adapt to the proposed model to observe their performance and compare with our proposed algorithm. Extensive experiments are conducted on both real-world and synthetic workloads to evaluate the accuracy and efficiency of the algorithms. The experimental results demonstrate that DCBB delivers near-optimal solutions with a convergence rate that is much faster than those of the other search-based algorithms evaluated. In particular, DCBB yields the optimal solution for a real-world workload with an execution time that is an order of magnitude shorter than that required by the original branch-and-bound (BB) algorithm.
	\end{abstract}

	\section{Introduction}
	\label{intro}
	Cloud computing is a computing paradigm that enables convenient, measurable, and on-demand network access to a pool of configured physical resources, such as CPU and memory. It can be categorized into three major deployment models: public clouds, private clouds and hybrid clouds \cite{mell2011nist}. Although public clouds still occupy the largest portion of the total cloud infrastructure, private clouds are attracting increasing attention from both industry and academia \cite{cloud2017idc} because of their better security and privacy control. According to a 2017 survey \cite{kim2017cloud} focusing on the adoption of cloud computing among IT professionals, 95\% of respondents used cloud platforms, and 75\% used private clouds or hybrid clouds. Moreover, previous studies \cite{kim2017cloud,goyal2014public} have revealed that a high upfront cost is among the most critical challenges associated with private clouds. Thus, there is a demand for efficient resource management methods those can reduce the scale of datacenters in order to popularize private cloud computing.\par
	Although the existing resource management methods have been well exploited, most of them are designed for general or public clouds. There is a need for more research based on the distinctive characteristics of private clouds, such as their predictable workloads and limited resources, to develop more efficient resource scheduling methods for the private cloud environment. In recent years, researchers have proposed multiple resource management methods, including task allocation \cite{ficco2017optimized,ramanathan2018towards} and workflow scheduling \cite{ye2017hybrid} methods, for achieving various goals specifically in the private cloud environment.\par
	The motivation of this work is to propose efficient virtual machine placement (VMP) methods for private clouds, in which resources are more limited and the workloads are more predictable than those of public clouds. The aim is to reduce the high upfront cost of datacenters, which is a key barrier to the popularization of private clouds. To achieve this goal, this work focuses on minimizing the number of servers (\#servers), which can also contribute to energy efficiency. Reducing the \#servers is one of the most straightforward and efficient methods of reducing the upfront cost of a private cloud since it can directly lower the costs of site use, server purchase, refrigeration, etc. Since resources are relatively limited in private clouds, we employ advance reservation \cite{toosi2015revenue,de2017bare} to increase the resource utilization ratio and reduce resource contention. \par
	VMP is a critical resource management method for cloud computing to improve performance, lower resource consumption and reduce maintenance cost \cite{masdari2016overview}. Many VMP methods have been proposed with various objectives, including effective load balancing, high energy efficiency, and high network traffic efficiency. However, the private cloud environment, in which resources are more limited and workloads are more predictable, has received little attention. As private clouds receive increasing interest from industry and academia, one of the emerging challenges of VMP is to determine how to conduct efficient scheduling to minimize the \#servers and thus reduce the high upfront cost in the private cloud environment.\par
	
	Although the majority of research and industry applications still focus on on-demand provision, advance reservation has been attracting increasing interest in the literature. Similar to the widely adopted appointment system \cite{feldman2014appointment}, the advance reservation approach can improve the scheduling efficiency and mitigate resource contention by making use of additional time information. Applications of advance reservation in cluster computing \cite{irwin2006sharing,lawson2002multiple} and grid computing \cite{elmroth2009standards,farooq2005impact} have been extensively researched to exploit its potential. In recent years, researchers have applied advance reservation in cloud computing to improve energy efficiency \cite{de2017bare} and maximize revenue \cite{toosi2015revenue,chase2017joint}. Moreover, cloud providers (e.g., Amazon \footnote{https://aws.amazon.com/}) have also provided reserved instances to satisfy user requirements. Because of the more predictable workloads and the resource limitations in private clouds, advance reservation can be effectively employed to increase the resource utilization ratio and reduce resource contention.
	
	Bin packing approaches are typically employed to address VMP problems. However, classic bin packing concentrates only on resources and ignores time information, which makes it difficult to address problems with an additional time dimension (e.g., advance reservation). Compared to classic bin packing, dynamic bin packing (DBP) can better handle VMP problems involving reservations since it considers time factors. By definition, DBP \cite{coffman1983dynamic} aims to model scenarios in which items arrive and depart randomly. DBP can be further classified into clairvoyant and nonclairvoyant settings depending on when the scheduler becomes aware of the departure times of virtual machines (VMs). Initially, researchers focused on nonclairvoyant dynamic bin packing (NCDBP), in which the system does not know the departure times of VMs until they have departed. However, with advances in workload prediction techniques \cite{park2017runtime,calheiros2015workload,gandhi2012hybrid}, clairvoyant dynamic bin packing (CDBP), in which the system becomes aware of the departure times of VMs when they arrive, has received increasing attention in recent years. Although efforts have been made to apply the DBP model in cloud computing, few studies have been conducted that have considered a heterogeneous environment or multidimensional resources, and this research gap impedes the further application of DBP in this context.\par
	
	The present research is applicable to the following scenario:
	\begin{itemize}
		\item{The employees of an organization need to use computing resources to support their work.}
		\item{The workloads are reasonably predictable and stable with regard to the required amounts of resources and their periods of usage.}
		\item{The company is concerned with issues such as security and confidentiality and thus prefers a private cloud.}
		\item{The organization hopes to minimize the datacenter size to reduce the upfront cost of building its own datacenter.}
	\end{itemize}

	The main contributions of this paper are as follows:
	\begin{enumerate}
		\item  A novel model and algorithm are proposed to reduce the upfront cost, which is the main barrier to the popularization of private clouds, by reducing the total \#servers required.
		\item  A formal definition of the enhanced CDBP problem with a heterogeneous environment and multidimensional resources is presented to better address VMP problems with an additional time dimension.
		\item  A novel branch-and-bound algorithm with a divide-and-conquer strategy (DCBB) is proposed to deliver near-optimal scheduling solutions within an execution time that is significantly shorter than those required by the other search-based algorithms evaluated.
		\item The previously proposed ant colony system with an order exchange and migration technique (OEMACS) is enhanced by endowing it with the ability to handle heterogeneous environments, multidimensional resources, and additional time information, thus making the algorithm more practical.
		\item  Various algorithms are analyzed, evaluated, and compared from different perspectives on real-world and synthetic workloads.
	\end{enumerate}
	
	A list of common acronyms used throughout this paper is presented in~\Cref{table:abbrGeneral} for the reader's convenience. \par
	\begin{table}
		\centering
		\caption{List of the main acronyms used in this paper}
		\label{table:abbrGeneral}       
		\begin{tabular}{ll}
			\hline\noalign{\smallskip}
			Acronym	& Definition                                    \\
			\noalign{\smallskip}\hline\noalign{\smallskip}
			\#servers	& Number of servers                            \\
			\#VMs	& Number of virtual machines                 \\
			BB	&   Branch-and-bound                                  \\
			\C   & Clustered set of virtual machines					 \\
			CDBP	& Clairvoyant dynamic bin packing                \\
			DBP	    & Dynamic bin packing                            \\
			DCBB	& Branch-and-bound algorithm with a divide-and-conquer strategy             \\
			DDFF	& Duration-descending first fit                  \\
			DDFF$^{+}$	& Duration-descending first fit with a shuffling process                 \\
			FF	& First fit                                          \\
			FF$^{+}$	& First fit with a shuffling process                                         \\
			MGC & Most-greedy clustering					 \\
			NCDBP & Nonclairvoyant dynamic bin packing              \\
			OEMACS	&  Ant colony system with an order exchange and migration technique\\
			OEMACS$^{+}$	&  Time-aware and multidimensional OEMACS\\
			\CSet    & Set of clustered sets of virtual machines		 \\
			VM	& Virtual machine                                    \\
			VMP	& Virtual machine placement                          \\
			\noalign{\smallskip}\hline
		\end{tabular}
	\end{table}
	
	The remainder of this paper is organized as follows. \Cref{sec:related works} first introduces related work. Then, \Cref{sec:system model} explains the system model. Next, \Cref{sec:scheduling algorithms} presents the scheduling algorithms, and \Cref{sec:implementation and experiments} describes the implementation and experiments. Finally, \Cref{sec:conclusion and future works} concludes the paper. 
	\section{Related Work}
	\label{sec:related works}
	In this paper, CDBP is applied in VMP to enhance the classic VMP model with an additional time dimension, corresponding algorithms are designed to address the modified problem, and the proposed methods are analyzed. In this section, the existing methods those focus on VMP and CDBP models are introduced and discussed. 
	\subsection{Virtual Machine Placement}
	VMP, an essential process for the initial placement of new VMs, has been extensively investigated in the literature \cite{masdari2016overview,usmani2016survey,panigrahy2011heuristics,gao2013multi,tang2015hybrid} on cloud computing resource management. The goal of this process is to initially allocate VMs to servers based on certain objectives, including energy conservation \cite{fard2017dynamic,zheng2016virtual,xiao2015solution}, cost minimization \cite{vu2014traffic,kanagavelu2014virtual}, resource saving \cite{gupta2018resource,liang2014placement,sayeedkhan2014virtual}, and load balancing \cite{xu2017survey}. \par
	Researchers have applied numerous algorithms to achieve efficient VMP. Accurate algorithms such as linear programming \cite{anand2013virtual}, stochastic integer programming \cite{chaisiri2009optimal}, and pseudo-Boolean optimization \cite{ribas2012modelling} have been studied to provide optimal scheduling solutions. Despite their accuracy, optimal algorithms are computationally prohibitive since VMP is well known to be an NP-hard problem \cite{panigrahy2011heuristics}. To accelerate the scheduling process, many heuristic algorithms based on a first-fit (FF) strategy \cite{panigrahy2011heuristics,fang2013power}, a best-fit strategy \cite{fang2013power,dong2013virtual}, a worst-fit strategy \cite{fang2013power} or a first-come-first-served strategy \cite{moreno2013improved} have been proposed to reduce the execution time, at the cost of some decrease in accuracy. With recent advances in evolutionary algorithms, researchers have also applied algorithms such as the frog leaping algorithm \cite{luo2014hybrid}, ant colony optimization \cite{gao2013multi,liu2016energy}, and genetic algorithms \cite{tang2015hybrid} for VMP to improve the scheduling performance. In 2016, Liu et al. \cite{liu2016energy} proposed OEMACS, an ant colony system with an order exchange and migration technique, which addresses VMP problems more effectively than other evolutionary and traditional algorithms do. \par
	Although extensive studies have been conducted in the field of VMP, only a small number of these studies have focused on private clouds, in which the workloads are more predictable and resources are more limited. Researchers have applied the genetic algorithms~\cite{quang2013genetic,agrawal2015power} and an artificial bee colony algorithm~\cite{agrawal2015power} to address VMP problems in private clouds with a focus on power efficiency, but these studies did not consider the distinctive characteristics (e.g., predictable workloads and limited resources) of private clouds to improve their performance. To better handle such scenarios, a formal representation of the VMP problem combined with CDBP is presented in this paper, and efficient algorithms are proposed to handle this problem effectively and efficiently. In addition, several VMP algorithms designed for the classic model, including FF and OEMACS, are adapted for use within our proposed heterogeneous and multidimensional CDBP model to observe their performance and enable comparisons with the proposed algorithm.
	\subsection{Clairvoyant Dynamic Bin Packing}
	Resource-aware VMP has typically been abstracted into a bin-packing problem that consists of a situation in which several items need to be packed into the minimum number of bins \cite{shi2013provisioning}. Bin packing and its $d$-dimensional variants have been extensively studied \cite{coffman2013bin,de1981bin,bansal2006bin,han1994multiple} since the 1960s. Many approximation algorithms have been proposed for $1$-dimensional bin packing \cite{coffman2013bin}. Fernandez de La Vega and Lueker \cite{de1981bin} proposed the first polynomial-time approximation scheme for $1$-dimensional bin-packing problems and proved that no such polynomial-time approximation scheme is possible for $2$-dimensional packing problems. In addition to the commonly considered case of homogeneous bins, several researchers have proposed algorithms for bin-packing problems in heterogeneous environments \cite{han1994multiple}. Although classic bin packing has been extensively employed to model resource-aware VMP, it encounters difficulties in describing time-enhanced cases, e.g., advance reservation.\par
	DBP \cite{coffman1983dynamic} is an extension of classic bin packing that additionally considers arrival time and duration, with items arriving and departing dynamically. Compared to classic bin packing, DBP can better model the advance reservation scenario and result in more efficient scheduling with time multiplexing. When DBP was first proposed and analyzed by Coffman et al. \cite{coffman1983dynamic} for allocation problems in computer systems, they focused on the NCDBP case, in which the scheduler does not know the departure times of VMs until they depart. Coffman et al. \cite{coffman1983dynamic} applied an FF strategy to reduce the \#servers required and proved that no online algorithm can obtain a performance bound that is lower than the FF bound. Later, researchers applied this model to reduce the total server usage time. Li et al. \cite{li2014dynamic} proved that the upper bound on the competitive ratio achieved with the FF strategy is 2$\mu$ + 13 and that the competitive ratio of the best fit is not bounded for $\mu$. They then proposed a modified FF strategy to improve the competitive ratio to $\mu$ + 8 when $\mu$ is known. Subsequently, Kamali et al. \cite{kamali2015efficient} improved the upper bound on the competitive ratio to $2\mu$ + 1, and Tang et al. \cite{tang2016first} reduced the value to $\mu$ + 4. \par
	In recent years, researchers have paid more attention to the application of CDBP to minimize the total usage time of servers. In contrast to the nonclairvoyant model, in CDBP, the scheduler can perceive the departure time of a VM upon its arrival, which enables more flexible scheduling. Ren et al. \cite{ren2016clairvoyant} proposed the duration-descending first fit (DDFF) algorithm, with an approximation ratio of 5, and the dual-coloring algorithm, with an approximation ratio of 4, as offline solutions. In 2017, Azar and Yossi \cite{azar2017tight} proposed a classify-by-duration FF strategy with a competitive ratio equal to the lower bound on the competitive ratio $\sqrt{\log\mu}$ of any online algorithm. \par
	The DBP model, which enables more efficient and flexible resource scheduling using the additional time dimension, seems promising for application in the private cloud environment, in which workloads are predictable and controllable. However, despite the great efforts researchers have directed toward DBP, their contributions have remained limited to homogeneous environments and $1$-dimensional resources to simplify the work. Moreover, as shown above, most research on DBP has sought to minimize the usage time of all servers. In this paper, a heterogeneous and multidimensional CDBP model and DCBB algorithm are proposed that can handle heterogeneous environments and multidimensional resources in order to minimize the total \#servers required. 
	\section{System Model}
	\label{sec:system model}
	In this section, a novel heterogeneous and multidimensional CDBP model is presented for the VMP problem in private clouds, in which the workloads are predictable and resources are limited. This model aims to better characterize the real-world VMP problem by providing a more detailed description of resources and time factors. In addition, with the additional arrival and duration information provided by the model, the scheduler can perform more efficient scheduling through time multiplexing. To provide a formal representation of the model, the VMs and servers are first defined; then, the time-enhanced constraints and objectives are clarified; and finally, the presented model is analyzed.  \par
	Let $S = (s_1,s_2,...,s_m)$ and $V = (v_1, v_2,..., v_n)$ denote the set of servers and the set of VMs, respectively. The VM $v_j$ in $V$ consists of a triple $(a_j, p_j, \vec{r^v_j})$, where $a_j$ is the arrival time, $p_j$ is the usage duration, and $\vec{r^v_j}$ represents the resources that $v_j$ demands. Thus, $v_j$ represents that a VM demanding $\vec{r^v_j}$ resources arrives at time $a_j$ and remains for a period of $p_j$. It is assumed that $a_j \geq 0$ and $p_j \ge 0$ for all $j$. Regarding servers, each server $s_i$ in $S$ can be simply represented by its resources $\vec{r^s_i}$ since it does not need an additional time dimension. Given that $l$ types of resources in total are considered, the resources associated with the server $s_i$ and the VM $v_j$ can be represented as $\vec{r^s_i}=(r^s_{i1},...,r^s_{il})$ and $\vec{r^v_j}=(r^v_{j1},...,r^v_{jl})$, respectively. Moreover, for each VM $v_j$, there exists at least one server $s_i$ satisfying $r^s_{ik} \ge r^v_{jk},$ $\forall k \in {1,2,...,l}$. Then, the heterogeneous and multidimensional CDBP model for VMP can be presented as follows.
	\begin{align}
		\min \quad & \sum_{i=1}^m\max_{j=1}^{n} x_{ij}  &  &  \label{eq:lp1}\\[1em]
		s.t. \quad & \sum_{j=1}^{n} r_{jk}^{v}x_{ij}u_{jt} \leq{} r_{ik}^s  &  &  \forall i = 1,...,m, \forall k = 1,...,l \label{eq:lp2},\forall t = 0,...,T \\
		& \sum_{i=1}^{m} x_{ij} = 1                                    &  & \forall j=1,...,n \label{eq:lp3} \\
		& x_{ij} \in \{0,1\}                                             &  & \forall i =1,...,m, \forall j=1,...,n \label{eq:lp4} \\
		& u_{jt} \in \{0,1\}                                             &  & \forall j = 1,...,n, \forall t = 0,...,T  \label{eq:lp5}   	  
	\end{align}
	The symbols used in the formulae are explained in \Cref{table:symbolOfModel}.\par
	\begin{table}
		\centering
		\caption{Symbols used in the system model}
		\label{table:symbolOfModel}       
		\begin{tabular}{ll}
			\hline\noalign{\smallskip}
			Symbol	& Definition       \\
			\hline\noalign{\smallskip}
			$l$ & \text{Number of resource dimensions} \\ 
			$m$ & \text{\#servers} \\ 
			$n$ & \text{\#VMs} \\ 
			${{r}^s_{ik}}$ & \text{Amount of the }$k^{th}$\text{ resource possessed by the }$i^{th}$\text{ server} \\
			${{r}^v_{jk}}$ & \text{Amount of the }$k^{th}$\text{ resource demanded by the }$j^{th}$\text{ VM}\\
			$t$ & \text{An instant of time in the experimental period} \\ 
			$T$ & \text{Total time of the experiment} \\ 
			${{u}_{jt}}$ & A variable indicating \text{whether }\text{the execution time of the }$j^{th}$\text{ VM contains the time instant }$t${;} \\    
			& its value is 1\text{ if the execution time of the }$j^{th}$\text{ VM contains }$t$\text{ and is }0\text{ otherwise} \\  
			${{x}_{ij}}$ & A variable indicating \text{whether the }$j^{th}$\text{ VM is assigned to the }$i^{th}$\text{ server;} \\
			& its value is 1\text{ if the }$j^{th}$\text{ VM is assigned to the }$i^{th}$\text{ server and is 0 otherwise}  \\

			\noalign{\smallskip}\hline
		\end{tabular}
	\end{table}
    As \Cref{eq:lp1,eq:lp2,eq:lp3,eq:lp4,eq:lp5} indicate, the proposed model considers the uptime of the VMs, which enables more flexible and efficient resource scheduling. Furthermore, because it considers heterogeneous and multidimensional resources, the model can better reflect real-world scheduling problems. The objective is to minimize the total \#servers required ($i.e.$, minimize the datacenter scale), as shown in~\Cref{eq:lp1}. If required, the objective function can be modified based on the user requirements. The constraints given in \Cref{eq:lp2} indicate that, at any time, each server should have an amount of resources equal to or greater than the total resources demanded by all the VMs that it is accommodating. Specifically, the left-hand side of~\Cref{eq:lp2} represents the total amount of the $k^{th}$ resource demanded from the $i^{th}$ server by all VMs, while the right-hand side represents the total amount of the $k^{th}$ resource possessed by the $i^{th}$ server. Moreover, the constraints in \Cref{eq:lp3} ensure that every VM is scheduled to one and only one server; this indicates that all VMs should be accommodated and that migration is not considered in this model. The constraints in \Cref{eq:lp4} and \Cref{eq:lp5} represent the ranges of $x$ and $u$, respectively.\par
	\begin{figure}
		\centering
		\begin{subfigure}[b][6cm][b]{.5\textwidth}
			\centering
			\includegraphics[width=4.5cm,height=2.5cm]{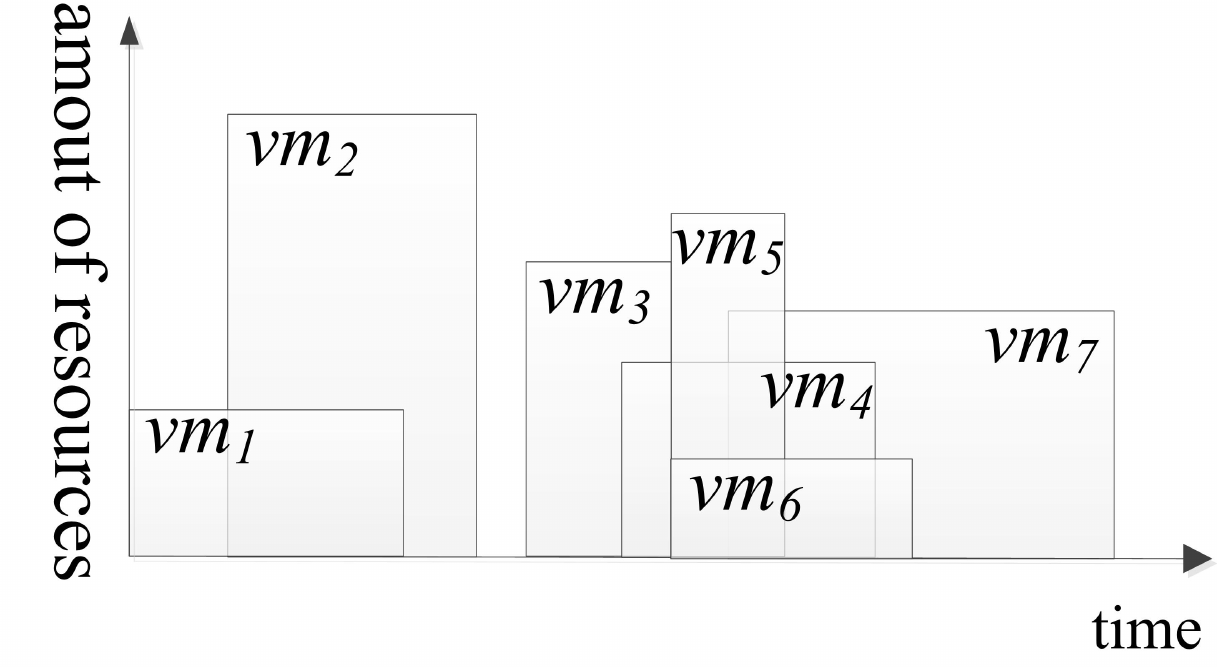}\\
			\caption{VM requests}
			\label{fig:schematic-1}
			\includegraphics[width=4.5cm,height=2.5cm]{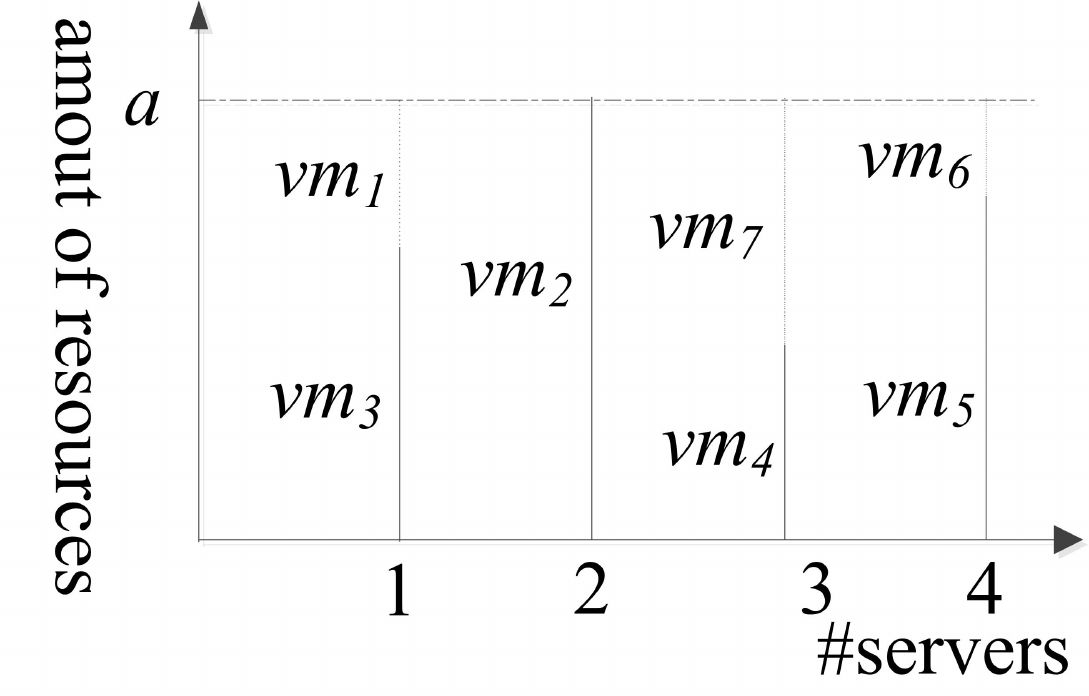}
			\caption{Solution under the classic model}
			\label{fig:schematic-2}
		\end{subfigure}%
		\begin{subfigure}[b][6cm][b]{.5\textwidth}
			\centering
			\includegraphics[width=5.5cm,height=5.4cm]{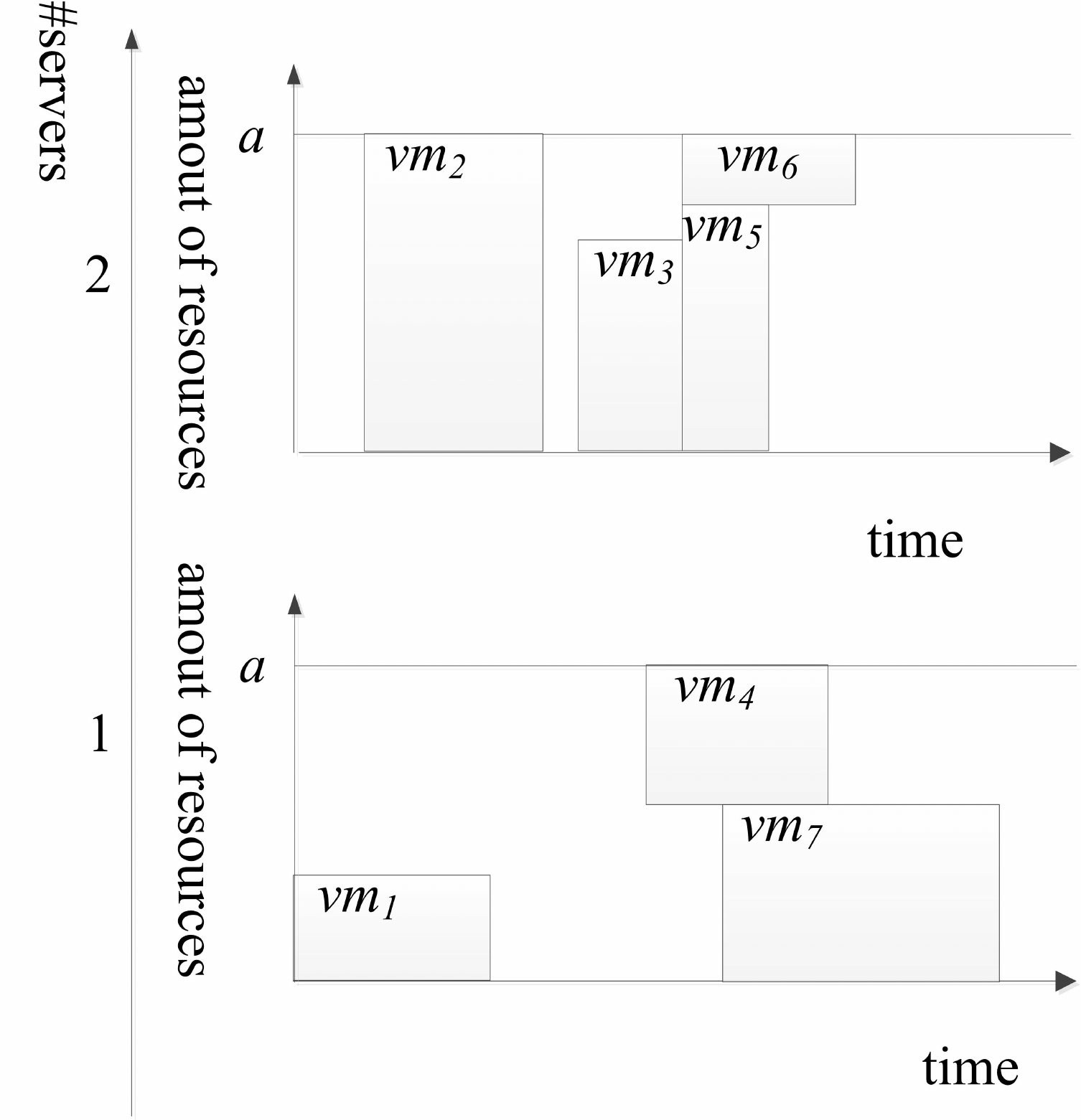}
			\caption{Solution under the CDBP model}
			\label{fig:schematic-3}
		\end{subfigure}
		\caption{Schematic illustrations of the classic and CDBP models}
		\label{fig:schematic}
	\end{figure}
	To provide an intuitive description of the proposed model, the allocation results obtained in the classic setting and in the CDBP setting are presented in \Cref{fig:schematic-2} and \Cref{fig:schematic-3}, respectively, given the VM requests in \Cref{fig:schematic-1}. In \Cref{fig:schematic}, rectangles are used to represent the VMs, with the height representing the amount of resources, the width representing the duration, and the horizontal position representing the arrival time. To clearly visualize the allocation results, $a$ is used to denote the amount of resources possessed by each server, as shown in \Cref{fig:schematic-2,fig:schematic-3}. Note that resources can have multiple dimensions (e.g., CPU, memory, and SSD) and that the servers can be heterogeneous in the present model, although only one dimension is used to represent the resources in \Cref{fig:schematic} to make the figure simpler and more concise. As shown in \Cref{fig:schematic-2}, under the classic model, the scheduler must allocate new resources for each VM because of the lack of time information. By contrast, as \Cref{fig:schematic-3} indicates, the CDBP model allows different VMs, for example, $vm_2$ and $vm_3$, to occupy the same resources in different periods to reduce the required \#servers. The resultant \#servers required to accept all requests is 4 in the classic setting and 2 in the CDBP setting. Thus, it can be concluded that the CDBP model can decrease the total \#servers required to accept all requests by means of time multiplexing. \par
	There are several special forms available for requests in the proposed model, as follows:\par
	\begin{enumerate}
		\item arrival time $= 0$: the request should be executed immediately, without a reservation.
		\item execution time $= \infty$: the duration of the request is undetermined, and the demanded resources should be reserved until it terminates.
		\item arrival time $= 0$ and execution time $= \infty$: the request will be degraded into a classic request since it does not provide any valid time information.  
	\end{enumerate}
	The use of these three special forms in the CDBP model is not recommended because they will reduce the degree of time multiplexing. In addition, from the third form, we find that the request under the CDBP model is more general than the original one, which indicates a wider range of application. \par
	As shown in~\Cref{eq:lp1}, the aim is to minimize the \#servers, thus reducing the upfront cost, while accepting all VMs. However, other objectives (e.g., load balancing and cost minimization) can also be adopted. \Cref{sec:scheduling algorithms} will present the algorithms proposed to address the VMP problem with an additional time dimension derived from the heterogeneous and multidimensional CDBP model introduced in this section. \par
	
	\section{Scheduling Algorithms}
	\label{sec:scheduling algorithms}
	To handle the problem derived from the model proposed in \Cref{sec:system model}, this section proposes DCBB algorithm and improves the state-of-the-art algorithm OEMACS and the classic algorithm DDFF. First, the motivations for and requirements of the algorithms for use within our proposed heterogeneous and multidimensional CDBP model are introduced. Then, we present the improved versions of DDFF and OEMACS, namely, DDFF$^{+}$ and OEMACS$^{+}$, that have been adapted for use within the proposed model. Finally, the DCBB procedure and its theoretical analysis are shown.\par
	\subsection{Motivations and Requirements}
	\label{subsed: motivations and requirements}
	Although various algorithms for the classic VMP problem have been proposed, as described in \Cref{sec:related works}, there is a need to design novel algorithms or improve existing algorithms to handle the additional time dimension in the CDBP model. Although the time dimension can be addressed in a time-sequential fashion using the classic online algorithms, their accuracies need to be improved, as shown in \Cref{sec:implementation and experiments}. Therefore, the DCBB algorithm is proposed to effectively and efficiently handle the VMP problem with the additional time dimension. In addition, DDFF and OEMACS are modified for use within the proposed model to observe their performance with additional time information and to compare them with the proposed DCBB algorithm.\par
	A practical VMP algorithm under the proposed model should satisfy the following requirements.
	\begin{enumerate}[start=1,label={\bfseries R\arabic*.},wide = 0pt, leftmargin = 2.2em]
		\item Multidimensional resources \cite{luo2014hybrid}: the algorithm should be able to handle resources with multiple dimensions, although some algorithms will become slower as the number of resource dimensions increases.
		\item Heterogeneity \cite{luo2014hybrid}: the algorithm should consider heterogeneous environments, in which servers have different amounts of resources, since such environments are a common feature of cloud datacenters.
		\item Time dimension \cite{gu2017reservation}: the algorithm should be able to handle the time dimension, which is the key feature of the CDBP model. Time multiplexing should be enabled to increase the resource utilization ratio and thus reduce the \#servers required to accommodate VMs.
		\item Availability \cite{toosi2015revenue}: the algorithm should ensure that resources are reserved in the appointed period for every accepted VM request, which is the basic requirement of the advance reservation mechanism.
	\end{enumerate} \par
	The following subsections present the improved algorithms DDFF$^{+}$ and OEMACS$^{+}$ and the proposed algorithm DCBB.
	
	\subsection{Duration-Descending First Fit with a Shuffling Process}
	In the literature, Runtian et al. \cite{ren2016clairvoyant} proposed the DDFF and dual-coloring algorithms to minimize the total usage time of servers under the CDBP model, with approximation ratios of 5 and 4, respectively. In this subsection, the aim is to modify the DDFF algorithm to minimize the \#servers in a heterogeneous environment with multidimensional resources. Although the dual-coloring algorithm has a lower approximation ratio, the difficulty of enhancing it to consider multidimensional resources impedes its application to our proposed model. DDFF first sorts the VMs in descending order by duration and then allocates the VMs in an FF manner. It can be easily adapted to a scenario with multidimensional resources because of its natural features. However, FF-based algorithms such as DDFF, which were originally designed for scenarios with 1 resource dimension, generally have difficulty sorting servers by their resources in a multidimensional-resource scenario since no inclusion relationship exists in this case. Inspired by \cite{liu2016energy}, DDFF$^{+}$ has been designed as an improved version of DDFF with an additional shuffling process to improve the scheduling performance. In addition, FF$^{+}$ has been designed using a similar improvement technique, although the detailed procedure is not presented here to avoid repetition. The pseudocode for DDFF$^{+}$ with a shuffling process is shown in \Cref{alg:ddff}.
	\begin{algorithm}
		
		\KwIn{a set of $n$ VMs, \VmSet; a set of $m$ servers, \ServerSet}
		\KwOut{an allocation of the VM requests in \VmSet to the servers in \ServerSet}

		\VmSet $\leftarrow \descendingSortByDuration(\VmSet)$ \hfill $O(nlogn)$\\
		\ServerSet $\leftarrow \shuffle(\ServerSet)$ \hfill  $O(m)$ \\
		$allocation \leftarrow \emptyset$   \hfill  $O(1)$\\
		\ForEach{ $vm$ in \VmSet}
		{\ForEach{ $server$ in \ServerSet}  
			{\If{$server.\canAccommodate(vm)$}
				{$allocation.\add(\{server,vm\})$ \hfill  $O(1)$}   
		}}
		
		\caption{DDFF$^{+}$}
		\label{alg:ddff}
		\Return{$allocation$}
	\end{algorithm} \par
	In Line 1 of \Cref{alg:ddff}, the VMs are sorted in descending order by duration, with the aim of improving the accuracy of the algorithm. In Line 2, the servers are shuffled to improve the scheduling performance in multidimensional-resource scenarios (\textbf{R1}) and heterogeneous environments (\textbf{R2}). The effectiveness of the sorting and shuffling processes has been demonstrated through experiments (\Cref{sec:implementation and experiments}). When the algorithm judges whether a server can accommodate a VM, as shown in Line 6, every resource dimension (\textbf{R1}) and the time dimension (\textbf{R3}) are simultaneously considered. Once a server that can accommodate the VM is found, the corresponding demanded resources will be reserved for the VM (\textbf{R4}), as shown in Line 7.

	Now that the details of DDFF$^{+}$ have been introduced, it can be proven that it is a polynomial-time algorithm with a time complexity of $O(nlogn)+ O(mn)$, where $m$ and $n$ are the \#servers and the number of VMs (\#VMs), respectively. Although this algorithm has a fast processing speed, its outcome is generally far from the optimum. To compensate for this degradation in accuracy, two more algorithms are presented below that are designed to achieve more accurate scheduling solutions.

	\subsection{Time-aware and Multidimensional OEMACS}
	As mentioned in \Cref{sec:related works}, the OEMACS algorithm \cite{liu2016energy} performs better than the conventional heuristics and other evolutionary algorithms when addressing the classic VMP problem in heterogeneous environments (\textbf{R2}). The local search techniques, namely, order exchange and migration, proposed by the authors for the ant colony system contribute to the impressive performance of OEMACS. Inspired by this algorithm, OEMACS$^{+}$ has been designed to consider the additional time dimension (\textbf{R3}) and the possibility of advance reservation (\textbf{R4}), allowing this algorithm to be applied to our proposed heterogeneous and multidimensional CDBP model. Furthermore, OEMACS$^{+}$ is also enhanced in terms of its consideration of multidimensional resources (\textbf{R1}), whereas OEMACS was originally designed for only two resource dimensions. To achieve the above goals, the majority of OEMACS was preserved, with the main modifications being concentrated in only a few formulae. The modified formulae are shown in~\Cref{eq:oemac1,eq:oemac2,eq:oemac3,eq:oemac4,eq:oemac5}, and the notations used are explained in \Cref{table:symboOEMACS}.
	\begin{table}
		\centering
		\caption{List of notations used to describe OEMACS$^{+}$}
		\label{table:symboOEMACS}       
		\begin{tabular}{ll}
			\hline\noalign{\smallskip}
			Notation	& Definition       														\\
			\noalign{\smallskip}\hline\noalign{\smallskip}											
			$D$		& The set of resource dimensions										\\
			$M_t$   & The set of servers utilized at time $t$	\\
			$P_{id}$	& The total amount of the $d^{th}$ resource possessed by the $i^{th}$ server                   \\
			$R_{itd}$ & The remaining amount of the $d^{th}$ resource of the $i^{th}$ server available at time $t$			\\
			$S$  & A solution to the VMP problem \\
			$S_b$ & The best solution in the current iteration \\
			$T(i)$  & The execution period of the $i^{th}$ VM							\\
			$v_{id}$	&  The amount of the $d^{th}$ resource demanded by the $i^{th}$ VM            \\
			$y_i$ & A variable indicating whether the $i^{th}$ server is used ($y_{i} = 1$) or not ($y_{i} = 0$) 	\\
			\noalign{\smallskip}\hline
		\end{tabular}
	\end{table}
	
	\begin{enumerate}
		\item The expression for identifying feasible servers is modified to ensure that the total amount of resources demanded by all VMs is not larger than the capacity of the target server in every resource dimension at any time, as shown in \Cref{eq:oemac1}.
		\begin{equation}
			\label{eq:oemac1}
			I{{'}_{j}}=\left\{ \left. i\left| {{R}_{itd}} \geq {{v}_{jd}}, 1\le i\le {{M}_{t}},\forall t\in T(j),\forall d\in D \right| \right\} \right.
		\end{equation}
		\item The expressions for the heuristic information (\Cref{eq:oemac2}), overload ratio (\Cref{eq:oemac3}), heuristic objective (\Cref{eq:oemac4}) and global pheromone updating operation (\Cref{eq:oemac5}) are improved by calculating the average remaining resource ratio during a time period considering all resource dimensions, as shown below.
		\begin{equation}
			\label{eq:oemac2}
			\eta '(i,j)=\frac{1-\frac{\sum\limits_{{{d}_{1}}\in D}{\sum\limits_{\begin{smallmatrix} 
							{{d}_{2}}\in D \\ 
							{{d}_{1}}\ne {{d}_{2}} 
							\end{smallmatrix}}{\left| \frac{1}{{{P}_{id{_{1}}}}}\left( \frac{\sum\limits_{t\in T(j)}{{{R}_{i}}_{t{{d}_{1}}}}}{|T(j)|}-{{v}_{j{{d}_{1}}}} \right)-\frac{1}{{{P}_{id{_{2}}}}}\left( \frac{\sum\limits_{t\in T(j)}{{{R}_{i}}_{t{{d}_{2}}}}}{|T(j)|}-{{v}_{j{{d}_{2}}}} \right) \right|}}}{\text{C}_{\left| D \right|}^{2}}}
			{\frac{\sum\limits_{d\in D}{\left| \frac{1}{{{P}_{id}}}\left( \frac{\sum\limits_{t\in T(j)}{{{R}_{i}}_{td}}}{|T(j)|}-{{v}_{jd}} \right) \right|}}{\left| D \right|}\text{+}1.0}
		\end{equation} \\
		
		\begin{equation}
			\label{eq:oemac3}
			over'(i)=\sum\limits_{d\in D}{\left| \frac{1}{{{P}_{id}}}\left( \frac{\sum\limits_{t\in T(j)}{{{R}_{i}}_{td}}}{|T(j)|}-v_{jd} \right) \right|}
		\end{equation} \\ 
		
		\begin{equation}
			\label{eq:oemac4}
			f{{'}_{2}}(S)=\sum\limits_{i\in {{M}_{t}}}^{{}}{\left( \sum\limits_{d\in D}{\left( \left| \frac{1}{{{P}_{id}}}\left( \frac{\sum\limits_{t\in T(j)}{{{R}_{i}}_{td}}}{|T(j)|} \right) \right| \right)}\cdot {{y}_{i}} \right)}
		\end{equation} \\
		
		\begin{equation}
			\label{eq:oemac5}
			\Delta \tau {{'}_{i}}=\frac{1}{{{f}_{1}}({{S}^{b}})}+\frac{1}{\sum\limits_{d\in D}{\frac{1}{{{P}_{id}}}\left( \frac{\sum\limits_{t\in T(j)}{{{R}_{i}}_{td}}}{|T(j)|} \right)}+1}
		\end{equation}
	\end{enumerate}
    
	Through the modifications shown in \Cref{eq:oemac1,eq:oemac2,eq:oemac3,eq:oemac4,eq:oemac5}, OEMACS$^{+}$ can be applied in our proposed heterogeneous and multidimensional CDBP model for enhancing the classic ant colony system with order exchange and migration as local search techniques. A brief explanation of the modified formulae is presented here, and the reader can refer to \cite{liu2016energy} for more details. First, \Cref{eq:oemac1} is used to select the feasible servers that have sufficient resources for the VM. Then, the heuristic information that is used to guide the greedy search is calculated using \Cref{eq:oemac2}. The overload ratio calculated in \Cref{eq:oemac3} represents the difference between the required and total resources after a VM has been accommodated. Then, the heuristic objective expressed in \Cref{eq:oemac4} is used to evaluate the solution. Finally, \Cref{eq:oemac5} is used to calculate the global pheromone, which guides the construction of better solutions.
	
	\subsection{Branch-and-Bound Algorithm with a Divide-and-Conquer Strategy}
	Although the classic BB algorithm can yield the optimal solution when applied to the linear programming model introduced in \Cref{sec:system model}, its execution time is beyond tolerance, especially in large-scale cases. In this subsection, we propose the DCBB algorithm, which is based on the BB algorithm and includes an additional divide-and-conquer process to improve the scheduling efficiency with little to no degradation in accuracy. To achieve this goal, DCBB first clusters the VMs into several VM sets, then works out the scheduling solutions for each set, and finally merges these subsolutions into the final one.

	The main DCBB procedure is as follows.
	\begin{enumerate}
		\item  Cluster the VMs into a set of clustered sets (\CSet) that satisfy the following conditions:
		\begin{enumerate}
			\item The execution times of every two VMs in the same clustered set (\C) overlap with each other.
			\item The execution times of any two VMs in different \Cs do not overlap.
		\end{enumerate}
		Then, place the VMs that cannot be clustered into any \C into the left set (\L).
		\item Schedule the VMs in different \Cs separately with BB.
		\item Schedule the VMs in \L with DDFF$^{+}$.
		\item Combine the solutions obtained in Steps 2 and 3.
	\end{enumerate}
	
	As shown above, rather than scheduling the VMs as a unit, the DCBB algorithm employs a divide-and-conquer strategy to handle the problem more efficiently. Step 1 enables time multiplexing (\textbf{R3}) by clustering the VMs into several VM sets based on time information. Then, these sets of VMs are scheduled separately in Steps 2 and 3, and finally, the subsolutions are merged without resource contention (\textbf{R4}) in Step 4. Although the original BB algorithm can be used to solve the VMP problem with multidimensional resources (\textbf{R1}) in heterogeneous environments (\textbf{R2}), it is computationally prohibitive. Through the divide-and-conquer process, DCBB achieves significantly improved efficiency  at the cost of a minor degradation in precision, as demonstrated in \Cref{sec:implementation and experiments}. \par
	\begin{algorithm}
		\KwIn{a set of VM requests, \VmSet; a set of servers, \ServerSet}
		\KwOut{an allocation of the VM requests in \VmSet to the servers in \ServerSet}
		\tcp{Cluster the VMs into the SCS and LS. A typical clustering algorithm is presented in Algorithm 3} (\cset, \l) $\leftarrow \cluster(\VmSet)$ 
		
		\ForEach{ \c in \cset}
		{$allocation \leftarrow allocation \cup \BB(\c,\ServerSet)$}
		
		$allocation \leftarrow allocation \cup \DDFF^{+}(\l,\ServerSet)$
		
		\caption{DCBB}
		\label{alg:dcbb}
		\Return{$allocation$}
	\end{algorithm}
	\begin{algorithm}[H]
		
		\SetKwFunction{findTimeMostVmsOccupied}{findTimeMostVmsOccupied}
		
		\KwIn{a set of VM requests, \VmSet}
		\KwOut{\CSet \cset; \L \l}

		\cset $\leftarrow \emptyset$, \l $\LA \emptyset$
		
		\While{$|\VmSet| \neq 0$}
		{
			\tcp{determine the time $t$ when the most VMs will be running} $t \LA \findTimeMostVmsContaining(\VmSet)$ \\
			\tcp{put all remaining VMs whose execution times contain $t$ into \c}    $\c \LA \getVMSetContainingTime(t, \VmSet)$ \\	
			\cset \LA \cset $\cup$ \c
			
			\VmSet \LA \VmSet \textbackslash \c\\
			
			\ForEach{ $vm_1$ in \VmSet}    
			{\ForEach{ $vm_2$ in \c}
				{\tcp{$\p(v)$ represents the execution time of VM $v$} \If{$\p(vm_1) \cap \p(vm_2) \neq \emptyset $ }
					{\l \LA $\l \cup \{vm_1\}$ \\
						\VmSet \LA \VmSet \textbackslash \{$vm_1$\}}
			}}
			
		}
		\caption{Most-Greedy Clustering Algorithm}
		\label{alg:clusterMost} 
		\Return{\cset, \l}
	\end{algorithm}
	The pseudocode for DCBB is shown in \Cref{alg:dcbb}. In this algorithm, Line 1 corresponds to the clustering process, while Lines 2 to 5 represent the processes of solving and merging. Various clustering algorithms have been designed such that the \Cs will satisfy the two conditions described in the DCBB procedure. However, in the current work, the different clustering algorithms perform similarly in both the theoretical analysis presented later in this subsection and the experiments we conducted. Thus, only one clustering algorithm, namely, most-greedy clustering (MGC), is presented here to demonstrate the process. The main strategy of MGC is to iteratively find the CS of the maximum size. The pseudocode presented in \Cref{alg:clusterMost} shows that MGC mainly involves the following steps:
	\begin{enumerate}
		\item Build a \C with the largest VM set in which the execution times of every two VMs overlap with each other. (Lines 3-5)
		\item Place all remaining VMs whose execution times overlap with that of any VM in a \C into the \L. (Lines 7-11)
		\item Repeat Steps 1 and 2 until all VM have been clustered into a set. (Lines 2-11)
	\end{enumerate}
	Steps 1 and 2 guarantee that the execution times of any two VMs in a \C will overlap with each other and that no two VMs in different \Cs will overlap, while Step 3 ensures that all VMs will be clustered into VM sets, based on which they will be scheduled later.
	
	\begin{figure}[hbt]
		\centering
		\includegraphics[width=10cm]{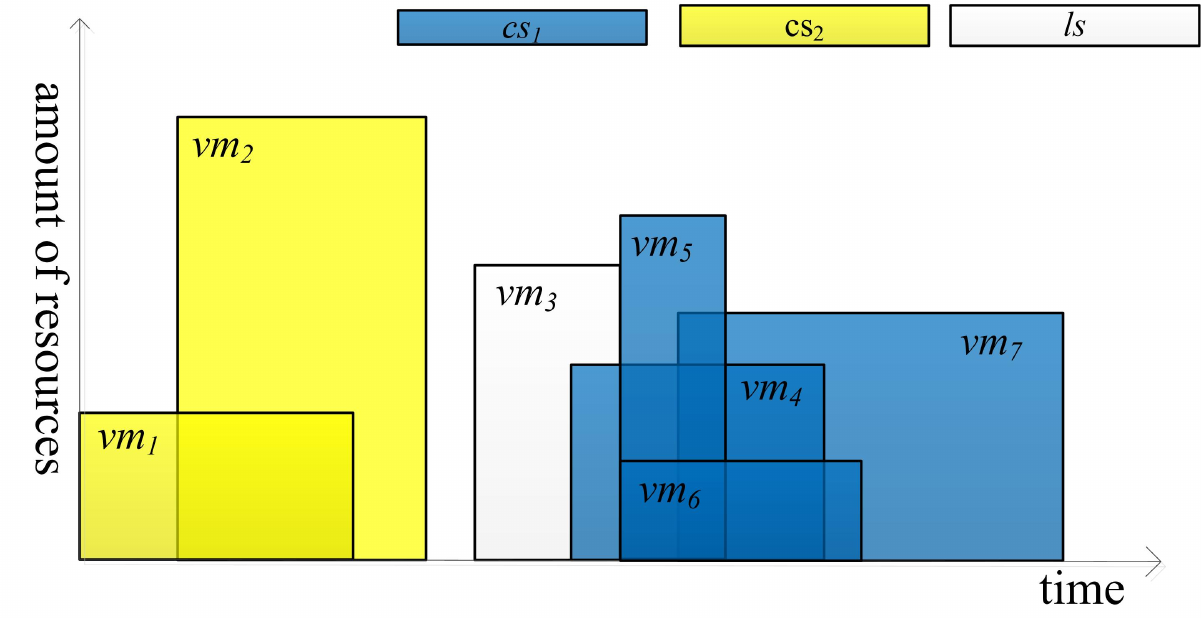}
		\caption{Schematic illustration of MGC}
		\label{fig:mgc}
	\end{figure}
	
	\Cref{fig:mgc} illustrates the clustering results obtained by MGC for 7 VMs. It shows that $vm_{4} – vm_{7}$ constitute the maximal set of overlapping VMs, the size of which is 4. Thus, MGC first puts  $vm_{4} – vm_{7}$ into $\c_1$. Then, $vm_{3}$ is clustered into the $ls$ since its execution time overlaps with that of $vm_{4}$. Finally, MGC clusters $vm_{1}$ and $vm_{2}$ into $\c_2$ since their execution times overlap.
	
	Although much more theoretical research on DCBB and suitable clustering algorithms needs to be conducted, some of the features of the clustering algorithm and its influence on DCBB have already been discovered, as described by the following lemmas and theorems, in which the abbreviations listed in \Cref{table:abbrDCBB} are used. 
	
	\begin{table}
		\centering
		\caption{List of abbreviations used to describe DCBB}
		\label{table:abbrDCBB}       
		\begin{tabular}{ll}
			\hline\noalign{\smallskip}
			Abbreviation	& Definition       \\
			\noalign{\smallskip}\hline\noalign{\smallskip}											
			A$(V)$	&  \#servers required by algorithm A for VM set $V$           \\
			\l       &  The LS generated by a clustering algorithm \\
			OPT    & The optimal algorithm                     \\
			p$(v)$	&  The run time of VM $v$										\\
			\cset & The \CSet generated by a clustering algorithm \\

			\noalign{\smallskip}\hline
		\end{tabular}
	\end{table}
	
	\begin{lemma}
		If one clustering algorithm yields $|\l|$ = 0, then all clustering algorithms will yield $|\l|$ = 0.
		\label{lemma:all 0}
	\end{lemma}
	\begin{proof}
		Assume that the lemma is invalid. Suppose that clustering algorithm $a_1$ yields \CSet $\cset_1=\{\c_{11}, \c_{12},...,\c_{1n}\}$ and that clustering algorithm $a_2$ yields \CSet $\cset_2= \{\c_{21}, \c_{22},...,\c_{2n}\}$ and \L $ls'$. Then, for any $v_i\in ls' $, there must exist two VMs $v_{k_1}$ and $v_{k_2}$ belonging to the same \C $\c_{2j}$ in $\cset_2$ that satisfy \Crefrange{eq:theory:1}{eq:theory:3}:
		\begin{equation}
		\label{eq:theory:1}
		\p(v_i) \cap \p(v_{k_1})= \emptyset,
		\end{equation}
		\begin{equation}
		\label{eq:theory:2}
		\p(v_i) \cap \p(v_{k_2})\ne \emptyset,\\
		\end{equation}
		\begin{equation}
		\label{eq:theory:3}
		\p(v_{k_1})\cap \p(v_{k_2})\ne \emptyset.\\
		\end{equation}
		From \Cref{eq:theory:2,eq:theory:3}, it can be inferred that $v_i$ and $v_{k_1}$ belong to the same $\c_{1k}$ since both $v_i$ and $v_{k_1}$ should belong to the same CS as $v_{k_2}$ does in $\cset_1$. However, according to \Cref{eq:theory:1}, $v_i$ and $v_{k_1}$ cannot be clustered into the same \C. Thus, the lemma is validated because the assumption fails.
	\end{proof}
	
	\begin{lemma}
		\label{lemma: 0 to optimal}
		If the clustering algorithm used in DCBB yields $|ls|$ = 0, then DCBB will produce the optimal solution.
	\end{lemma}
	\begin{proof}
		According to the independence of the execution times of the VMs in different \Cs and given $|\l|=0$,
		\begin{equation}
		\text{DCBB}(vmSet) = \text{DCBB} (\underset{\c_i \in \cset}{\bigcup}\c_i) = \underset{\c_i \in \cset}{max}\text{BB} (\c_i).
		\end{equation}
		Since
		\begin{equation}
		\text{BB}(\c_i) \leq \text{BB}(vmSet)\ \text{for}\ \forall \c_i \in \cset,
		\end{equation}
		the following holds:
		\begin{equation}
		\text{DCBB}(vmSet) \leq \text{BB}(vmSet).
		\end{equation}
		Since the BB algorithm is ideally accurate, it can be inferred that DCBB yields the optimal solution.
	\end{proof}

	\begin{theorem}
		\label{theorem: optima}
		If one clustering algorithm yields $|\l|=0$, then $DCBB$ integrated with any clustering algorithm will yield the optimal solution.
	\end{theorem}
	
	\begin{proof}
		Theorem \ref{theorem: optima} can be deduced based on Lemma~\ref{lemma:all 0} and Lemma~\ref{lemma: 0 to optimal}.
	\end{proof}

	\begin{lemma}
		If clustering algorithm $a$ yields $|\l|$ = 0 and the VMs $v_a$ and $v_b$ are clustered into the same \C $\c_{i}$ by algorithm $a$, then $v_a$ and $v_b$ will be clustered into the same \C by all clustering algorithms.
		\label{lemma:same C}
	\end{lemma}
	\begin{proof}
		Assume that $v_a$ and $v_b$ are clustered into different \Cs produced by a certain clustering algorithm. Then, $\p(v_a) \cap \p(v_b) = \emptyset$. However, it can be inferred that $\p(v_a) \cap \p(v_b) \ne \emptyset$ since $v_a$ and $v_b$ are both clustered into $\c_{i}$ by algorithm $a$, which leads to a contradiction.
	\end{proof}
	\begin{lemma}
		If one clustering algorithm results in $|\l|$ = 0, then all clustering algorithms will yield the same results.
		\label{lemma:same results}
	\end{lemma}
	\begin{proof}
		For any two clustering algorithms $a_1$ and $a_2$, if $a_1$ yields $\cset_1=\{\c_{11},$ $\c_{12},…,\c_{1m}\}$, then $a_2$ yields $\cset_2=\{\c_{21}, \c_{22},…,\c_{2n}\}$ without the LS because of Lemma~\ref{lemma:all 0}. Assume that $v_x$ belongs to $\c_{1i}$ and $\c_{2j}$. If $\c_{1i}$ is different from $\c_{2j}$, then a VM $v_y$ must exist such that either $v_y \in \c_{1i}$ and $v_y\notin \c_{2j}$ or $v_y \notin \c_{1i}$ and $v_y\in \c_{2j}$. Thus, $v_x$ and $v_y$ are clustered into the same \C only in $\cset_1$ or $\cset_2$, which contradicts Lemma~\ref{lemma:same C}.
	\end{proof}

	\begin{theorem}
		$\text{DCBB}(vmSet) \leq \text{OPT}(vmSet) + |\l|$.
		\label{theorem:upper bound}
	\end{theorem}
	\begin{proof}
		The proof can be divided into two separate cases:
		\begin{enumerate}[(i)]
			\item When $|\l| = 0$, $\text{DCBB}(vmSet) = \text{OPT}(vmSet)$ according to Lemma~\ref{lemma:all 0}.
			\item When $|\l| > 0$, let $vmSet' = vmSet - |\l|$. Since $vmSet'$ can be clustered into \Cs, 
			the following equation is satisfied according to Lemma~\ref{lemma:all 0}:
			\begin{equation}
			\begin{split}
			\text{DCBB}(vmSet') & = \text{OPT}(vmSet') \\
			& \leq \text{OPT}(vmSet) 
			\end{split}
			\end{equation}
			The worst case for DDFF($vmSet$) is $|\l|$ when each VM in $\l$ needs to be scheduled to a different server. Thus, it can be inferred that
			\begin{equation}
			\begin{split}
			\text{DCBB}(vmSet) & = \text{DCBB}(vmSet' \cup \l) \\
			& \leq \text{DCBB}(vmSet') + \text{DDFF}(\l) \\
			& \leq \text{OPT}(vmSet) + |\l|
			\end{split}
			\end{equation}
		\end{enumerate}
		Consequently, Theorem~\ref{theorem:upper bound} is satisfied for any $|\l|$.
	\end{proof}
	
	As shown above, no resource-related features are involved in the deductions of \Cref{lemma:all 0,lemma: 0 to optimal,lemma:same C,lemma:same results} and Theorem~\ref{theorem: optima}, which indicates that the conclusions satisfy requirements \textbf{R1} and \textbf{R2} well. Furthermore, Theorem~\ref{theorem: optima}, Lemma~\ref{lemma:all 0}, and Lemma~\ref{lemma: 0 to optimal} represent the typical conditions under which \textbf{R3} can be best fulfilled. Since all VMs can be clustered into several independent VM sets that do not overlap with each other, DCBB can achieve the optimal solution by merging the subsolutions for each subset under these conditions. Moreover, because all the VMs can be clustered into either a certain \C or the \L, based on which they will later be scheduled to a certain server, \textbf{R4} is not violated during the clustering process. In addition, Theorem~\ref{theorem:upper bound} presents the upper bound of our proposed DCBB algorithm.

	\section{Implementation and Experiments}
	\label{sec:implementation and experiments}
	In this section, experimental results are presented to evaluate and compare various algorithms, including BB, DCBB, OEMACS$^{+}$, DDFF$^{+}$, and FF$^{+}$, in our proposed heterogeneous and multidimensional CDBP model. The main metrics used for evaluation are the accuracy and execution time of each algorithm. As mentioned in \Cref{sec:related works}, the scheduling problem in the proposed model is NP-hard; thus, obtaining an optimal solution (the least \#servers required) in a reasonable time is computationally infeasible. Therefore, as widely adopted in the literature \cite{xiao2015solution,gupta2018resource,luo2014hybrid,liu2016energy,quang2013genetic}, we assess the accuracy on the basis of the \#servers, where a smaller \#servers indicates a higher accuracy. In the following, \Cref{subsec: workload design} first introduces the workloads. Then, \Cref{subsec:real world,subsec:convergence speed,subsec:shuffle,subsec:number,subsec:time} describe experiments conducted to answer the following questions: 
	
	\begin{enumerate}[start=1,label={\upshape\bfseries Q\arabic*:},wide = 0pt, leftmargin = 2.2em]
		\item How do the algorithms perform on a real-world workload? (\Cref{subsec:real world})
		\item What are the convergence rates of search-based algorithms, such as DCBB, BB, and OEMACS$^{+}$, under the proposed model? (\Cref{subsec:convergence speed})
		\item How does the shuffling process affect the performance of FF-based algorithms in a multidimensional environment? (\Cref{subsec:shuffle})
		\item How do the performances of the algorithms change with increasing problem scale?  (\Cref{subsec:number})
		\item How do time factors (i.e., the arrival times and durations of VMs) affect the performances of the algorithms? 	(\Cref{subsec:time})
	\end{enumerate}
	Finally, \Cref{subsec:summary} summarizes the experimental results.
	
	\subsection{Workloads}
	\label{subsec: workload design}
        A real-world workload is considered to observe the practical performances of the algorithms. In addition, synthetic workloads are generated to observe the influence of the \#VMs and time distribution on the algorithms. Furthermore, as shown in \Cref{table:workload}, 8 types of VMs were selected from the Amazon Elastic Compute Cloud (EC2) \footnote{https://aws.amazon.com/ec2/} to serve as workloads, and 3 types of servers were selected on the basis of the products available from Inspur Technologies Co., Ltd. \footnote{http://en.inspur.com/inspur/}, to make the experimental environment more similar to a real-world scenario. The selected types of both VMs and servers include CPU-intensive, memory-intensive and SSD-intensive representatives to cover a general set of cases.
	\begin{table}
		\centering
		\caption{Types of servers and VMs}
		\label{table:workload}       
		\begin{tabular}{*{4}{c}}
			\hline\noalign{\smallskip}
			&	 \#(V)CPUs & Memory (GB) & SSD (GB)             \\
			\noalign{\smallskip}\hline\noalign{\smallskip}
			\multirow{3}*{servers} &  16 &32  &160       \\
			& 8   &32  &160       \\
			& 8   &64  &320       \\
			\noalign{\smallskip}\hline\noalign{\smallskip}
			\multirow{8}*{VMs} & 1   &3.75&4         \\
			& 2   &7.5 &32       \\
			& 4   &15  &80       \\
			& 2   &3.75&32       \\
			& 4   &7.5 &80       \\
			& 8   &15  &160       \\
			& 2   &15.25&32       \\
			& 4   &30.5&80       \\
			\noalign{\smallskip}\hline
		\end{tabular}
	\end{table}
	
	In addition, a uniform distribution $U(a,b)$ and a Gaussian distribution $N(\mu,\sigma^2)$ are used to simulate the arrival times and durations, respectively of the VMs. 
	
	The details of the workloads are as follows.
	
	\begin{itemize}
		\setlength\itemsep{0.8em}
		\item{\textbf{Workload \RNum{1}: real-world workload}} \\
		Considering the completeness and quality of the workloads, we selected two real-world datasets, namely, ``RICC'' and ``UniLuGaia'', from the ``Logs of Real Parallel Workloads from Production Systems'' \cite{Feitelson2017workload} to evaluate and compare the accuracy and efficiency of the algorithms. In contrast to synthetic data, these real-world datasets have long time spans, sparse VM distributions, and occasionally incomplete records. Thus, data cleaning was first performed on the two datasets to improve the significance of the experiments. Considering the computational capacity of the experimental environment, 500 qualified records extracted from each dataset were used to conduct the experiments.
		\item{\textbf{Workload \RNum{2}: varying \#VMs and fixed time distributions}} \\
		Workload \RNum{2}, in which the total \#VMs varies from 24 to 336 while the distributions of the VM arrival times and durations are fixed, as shown in \Cref{table:distribution}, was generated to illustrate the influence of the \#VMs.
		\begin{table}
			\centering
			\caption{Distributions of the arrival times and durations of the VMs and their default parameters}
			\label{table:distribution}       
			\begin{tabular}{llll}
				\hline\noalign{\smallskip}
				Type	& Distribution  &Parameter $1$ &Parameter $2$   \\
				\noalign{\smallskip}\hline\noalign{\smallskip}
				Arrival time&Uniform&$0$ $(a)$ &  $240$ $(b)$              \\        
				Duration&	Gaussian&	$360$ $(\mu)$& $60$ $(\sigma)$      \\
				
				\noalign{\smallskip}\hline
			\end{tabular}
		\end{table}
		
		\item{\textbf{Workload \RNum{3}: fixed \#VMs and varying time distributions}}\\
		Workload \RNum{3} was designed to study the influence of time distributions on the algorithms. The total \#VMs of this workload is fixed at 160. For the time distributions, both the upper bound $b$ on the arrival times and the mean duration $\mu$ vary between 60 s and 420 s, while the lower bound $a$ on the arrival times and the variance $\sigma$ of the durations remain unchanged, as shown in \Cref{table:distribution}.  
	\end{itemize}

	The scheduling algorithms were evaluated and compared using the above workloads in a KVM-based VM with 8 VCPUs and 16 GB of memory. A private cloud platform was built using OpenStack \footnote{https://www.openstack.org/} to observe the performances of the proposed model and algorithms. A simulated environment was also established in which to conduct large-scale experiments. In the following sections, we do not differentiate the real and simulated experimental environments since they do not affect the scheduling results.
	
	\subsection{Experiment on the Real-World Workload}
	\label{subsec:real world}
	\begin{figure}
		\begin{subfigure}{0.5\textwidth}
			\centering
			\includegraphics[width=1\linewidth]{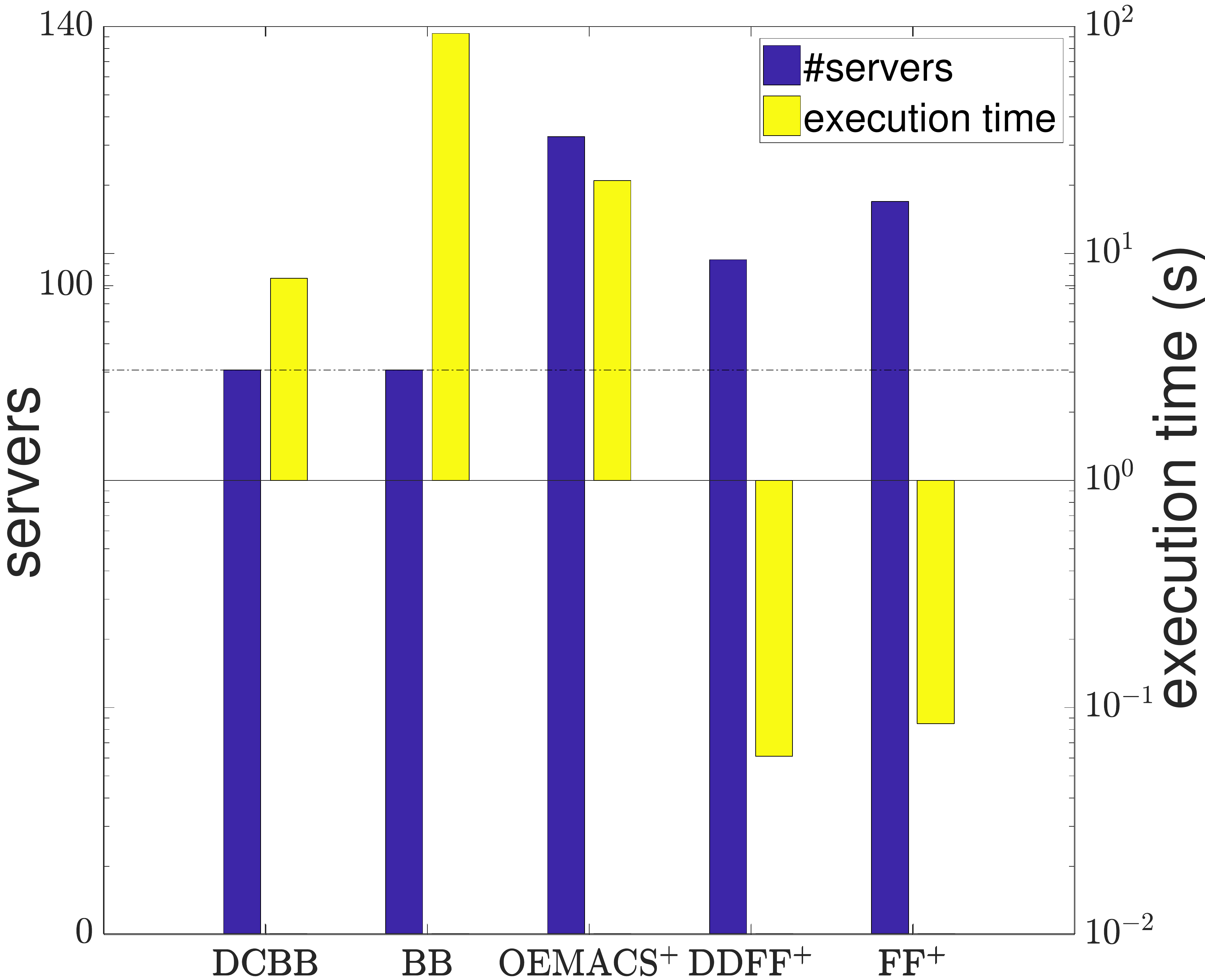}
			\caption{RICC dataset}
			\label{fig:real:1}
		\end{subfigure}%
		\begin{subfigure}{0.5\textwidth}
			\centering
			\includegraphics[width=1\linewidth]{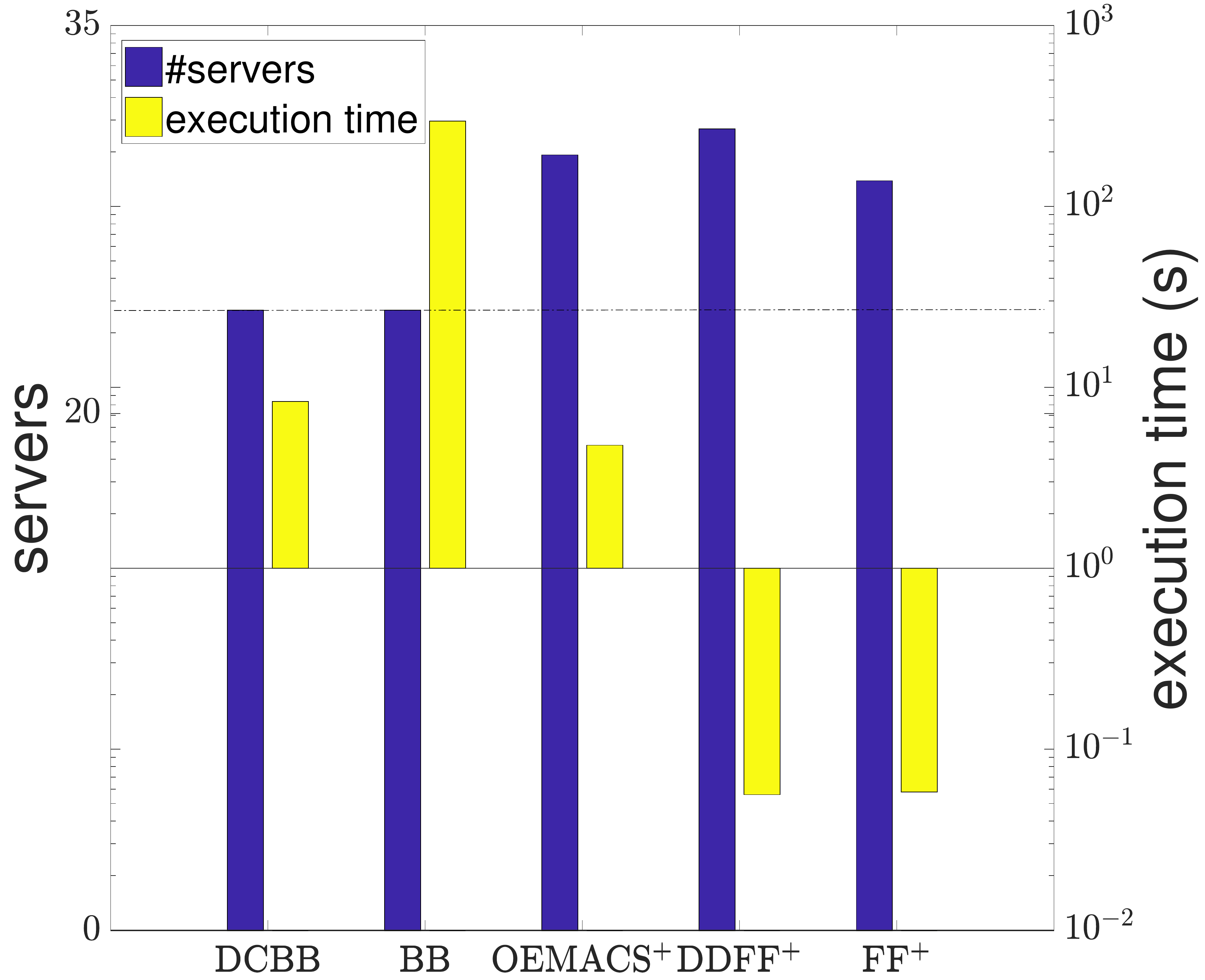}
			\caption{UniLuGaia dataset}
			\label{fig:real:2}
		\end{subfigure}
		\caption{Comparisons of the algorithms using real-world datasets. The fewest \#servers achieved by the algorithms are marked with dotted lines.}
		\label{fig:real}
	\end{figure} 
	In this subsection, Workload \RNum{1} is employed to check the performances of the algorithms on real-world datasets, which is a key component of the evaluation and comparison. The results are shown in \Cref{fig:real}, with dotted lines indicating the fewest \#servers required by the algorithms. \par
	First, the execution times are compared. As shown, DDFF$^{+}$ and FF$^{+}$ require less than 0.1 s to yield the solutions, which is much less than the times required by the other algorithms. The execution times of DCBB and OEMACS$^{+}$ are approximately dozens of seconds, whereas BB requires the most time – several hundreds of seconds. \par
	In terms of the \#servers required by each algorithm, DCBB and BB achieve the optimal results, requiring 19.46\% and 20.13\% fewer servers on average than the DDFF$^{+}$ and FF$^{+}$ algorithms do, respectively. DDFF$^{+}$ requires the third fewest servers on the RICC dataset; however, it has the worst accuracy on the UniLuGaia dataset. OEMACS$^{+}$ and FF$^{+}$ have accuracies similar to that of DDFF$^{+}$. \par
	To summarize, DCBB achieves the same optimal solution as BB does with an execution time that is an order of magnitude shorter. Moreover, OEMACS$^{+}$ requires nearly the largest \#servers with a relatively long execution time, which may be caused by the additional problem complexity introduced by the additional time and resource dimensions. Furthermore, the FF-based algorithms can produce a scheduling solution within a trivial execution time, indicating that they are suitable for real-time scheduling. In the following subsections, more comprehensive analyses of the algorithms will be presented based on synthetic data.
	
	\subsection{Convergence Rate Comparison}
	\label{subsec:convergence speed}
	As search-based algorithms, DCBB, BB, and OEMACS$^{+}$ can deliver better solutions given longer execution times, up to the time when the optimal solution is found. In particular, BB can theoretically always produce an optimal scheduling solution given enough time. However, the execution time cannot be arbitrarily long. Thus, the convergence rates of these algorithms should be studied to evaluate their performance and achieve a suitable compromise between accuracy and efficiency. In this subsection, DCBB, BB, and OEMACS$^{+}$ are applied to Workload \RNum{2} to compare the convergence rates of these algorithms. \par 
	
	\begin{figure}
		
		\begin{subfigure}{0.5\textwidth}
			\centering
			\includegraphics[width=0.8\linewidth]{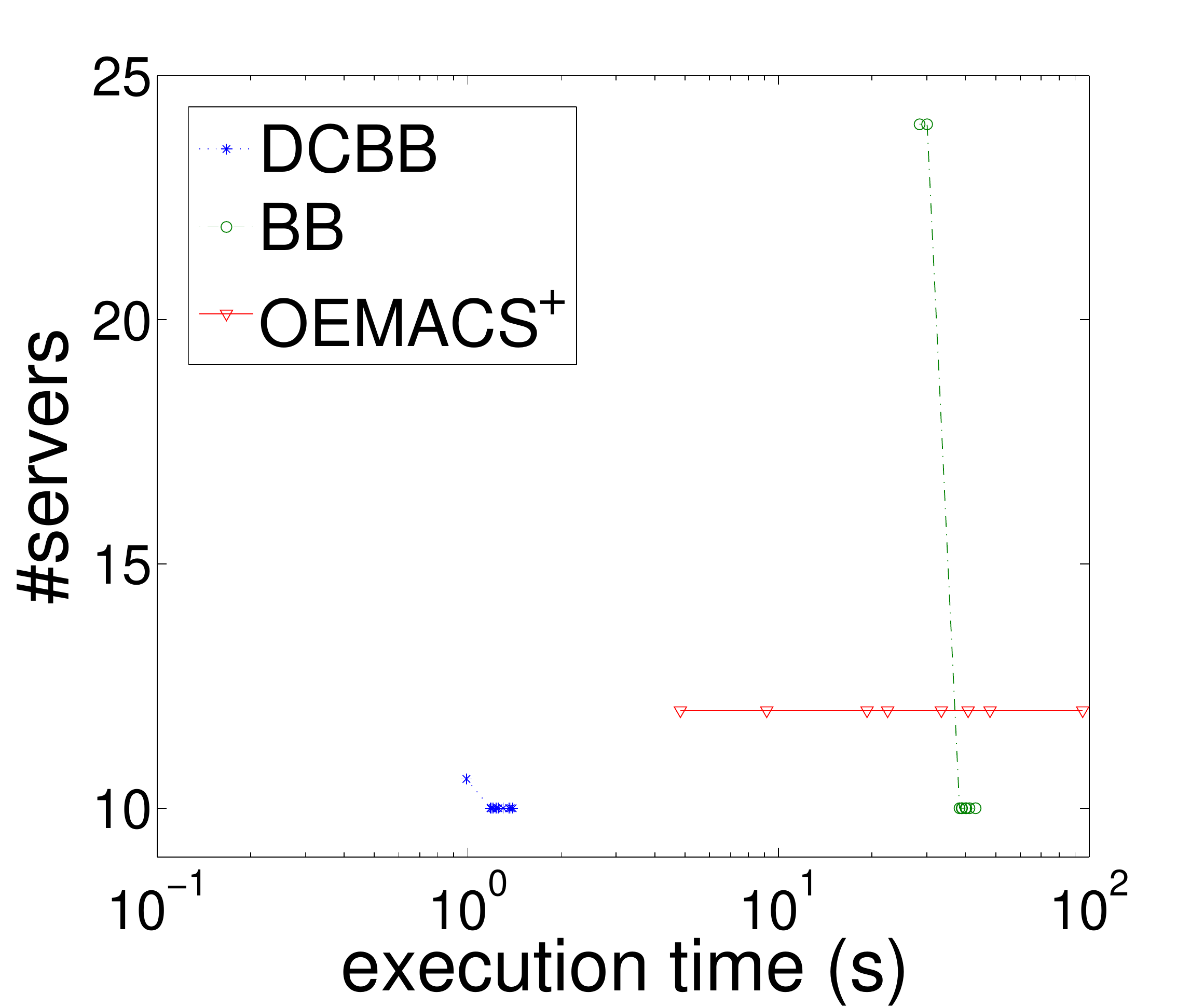}
			\caption{\#VMs = 24}
			\label{fig:conSpe:1}
		\end{subfigure}%
		\begin{subfigure}{0.5\textwidth}
			\centering
			\includegraphics[width=.8\linewidth]{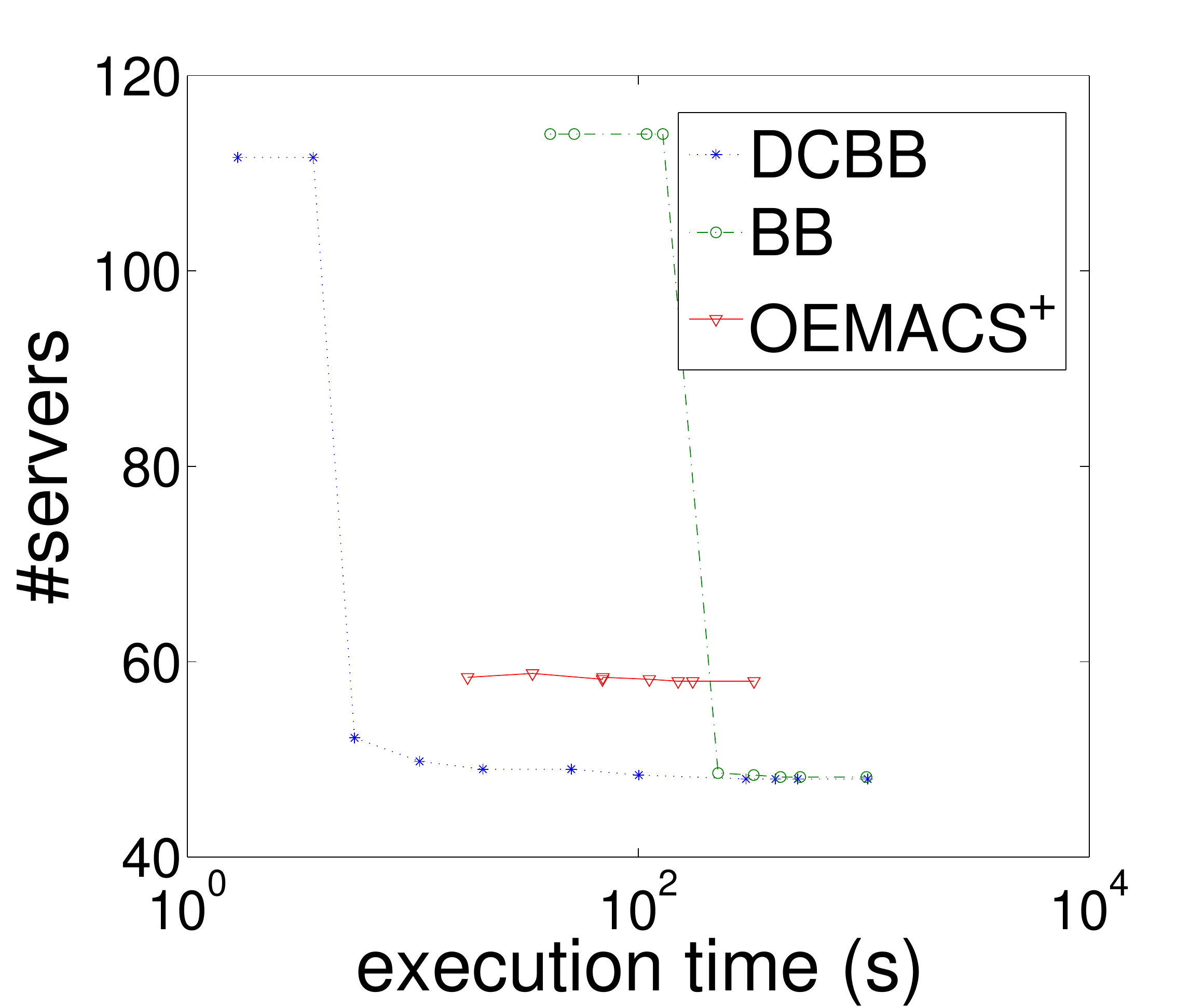}
			\caption{\#VMs = 120}
			\label{fig:conSpe:2}
		\end{subfigure}
		\bigskip
		\begin{subfigure}{0.5\textwidth}
			\centering
			\includegraphics[width=.8\linewidth]{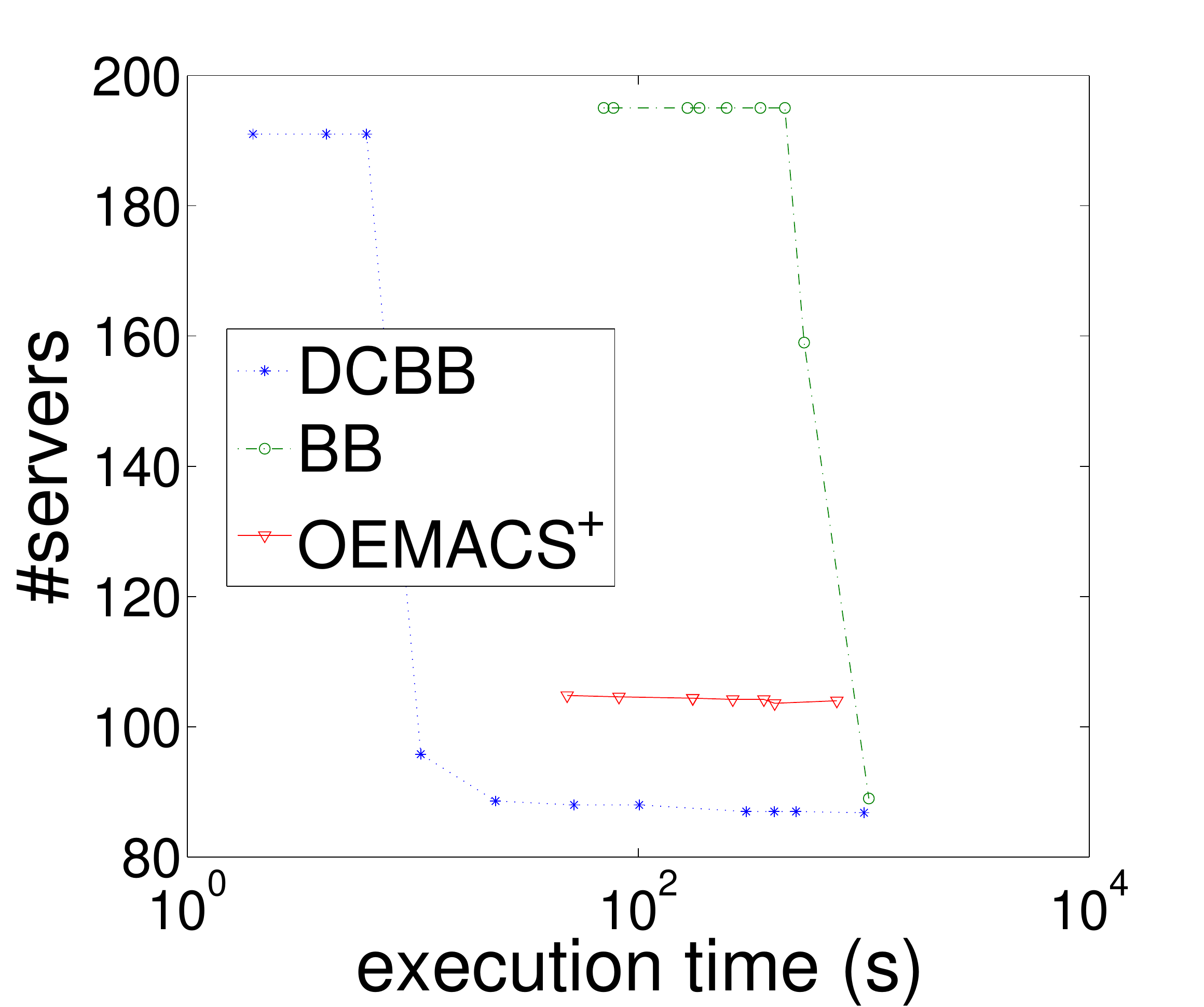}
			\caption{\#VMs = 216}
			\label{fig:conSpe:3}
		\end{subfigure}%
		\begin{subfigure}{0.5\textwidth}
			\centering
			\includegraphics[width=.8\linewidth]{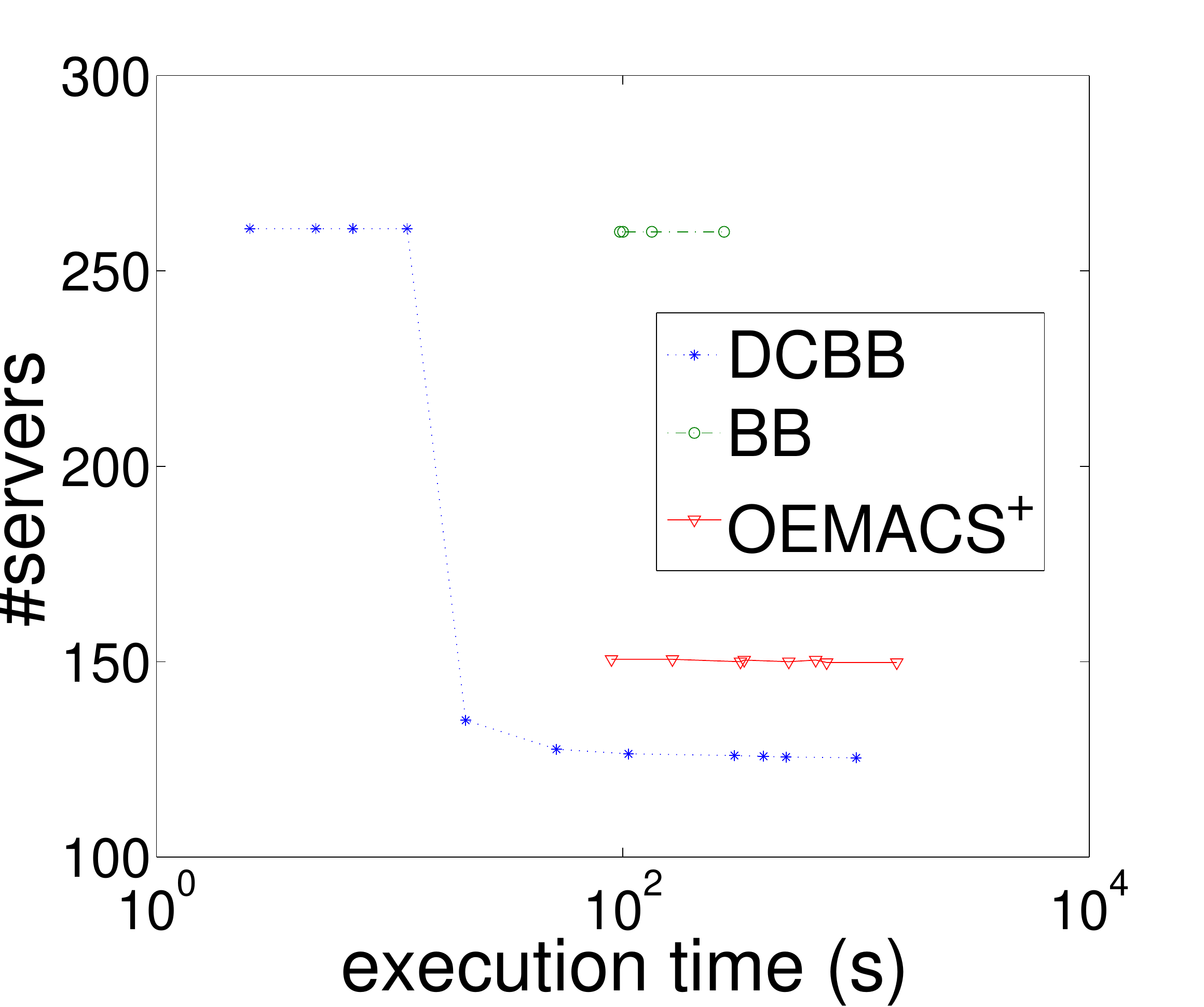}
			\caption{\#VMs = 312}
			\label{fig:conSpe:4}
		\end{subfigure}
		\caption{Comparison of the convergence rates}
		\label{fig:conSpe}
	\end{figure}
	\Cref{fig:conSpe} shows that the total \#servers required by DCBB decays exponentially with an increasing execution time, and thus, the convergence rate of this algorithm becomes the fastest. Furthermore, DCBB converges before 50 s in most cases. In contrast, the convergence rate of BB is much slower, with a nearly linear decay after a period of unchanging results. The missing data of BB after 50 s in \Cref{fig:conSpe:4}, where the \#VMs is 312, is caused by the excessive computational resource requirements of this algorithm. In the other cases shown in \Cref{fig:conSpe:1,fig:conSpe:2,fig:conSpe:3}, BB obtains nearly convergent results after 1000 s. For OEMACS$^{+}$, the \#servers required remains almost unimproved as the execution time increases. \par
	It can be concluded that DCBB achieves the fastest convergence rate, OEMACS$^{+}$ yields nearly unchanging results over time, and BB has the slowest convergence rate. In the following subsections, time limits of 50 s and 1000 s are set for DCBB and BB, respectively, and an iteration limit of 5 is set for OEMACS$^{+}$ to balance the accuracy and efficiency of these algorithms according to the results shown in \Cref{fig:conSpe}.
	
	\subsection{Effectiveness of Shuffling}
	\label{subsec:shuffle}
	In this subsection, the improved algorithms FF$^{+}$ and DDFF$^{+}$ are compared with the original algorithms FF and DDFF using Workload \RNum{2} to observe the effectiveness of the shuffling process.  \par
	\begin{figure}
		
		\begin{subfigure}{0.5\textwidth}
			\centering
			\includegraphics[width=0.9\linewidth]{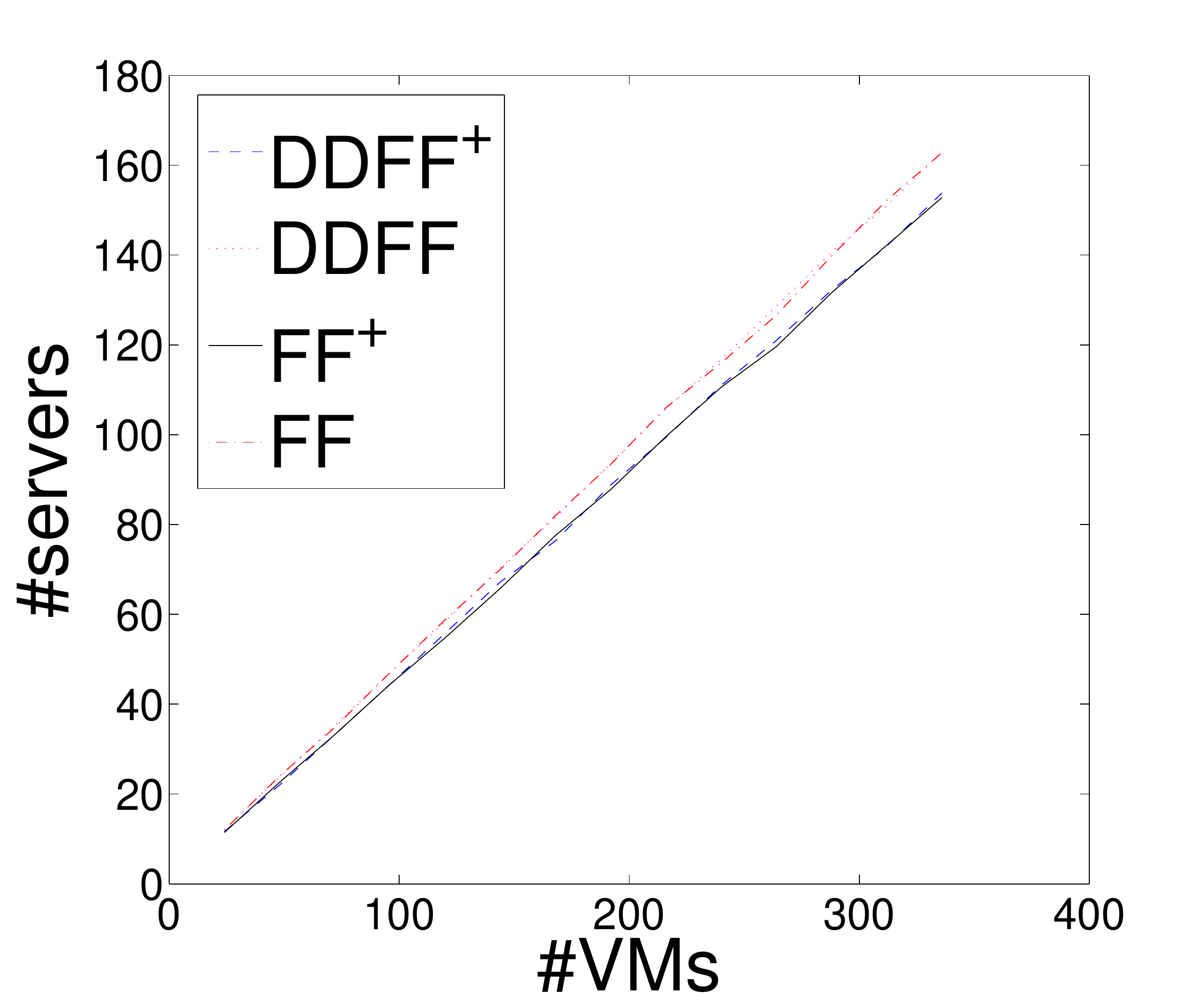}
			\caption{Effect of shuffling on \#servers}
			\label{fig:shuffle:1}
		\end{subfigure}%
		\begin{subfigure}{0.5\textwidth}
			\centering
			\includegraphics[width=0.89\linewidth]{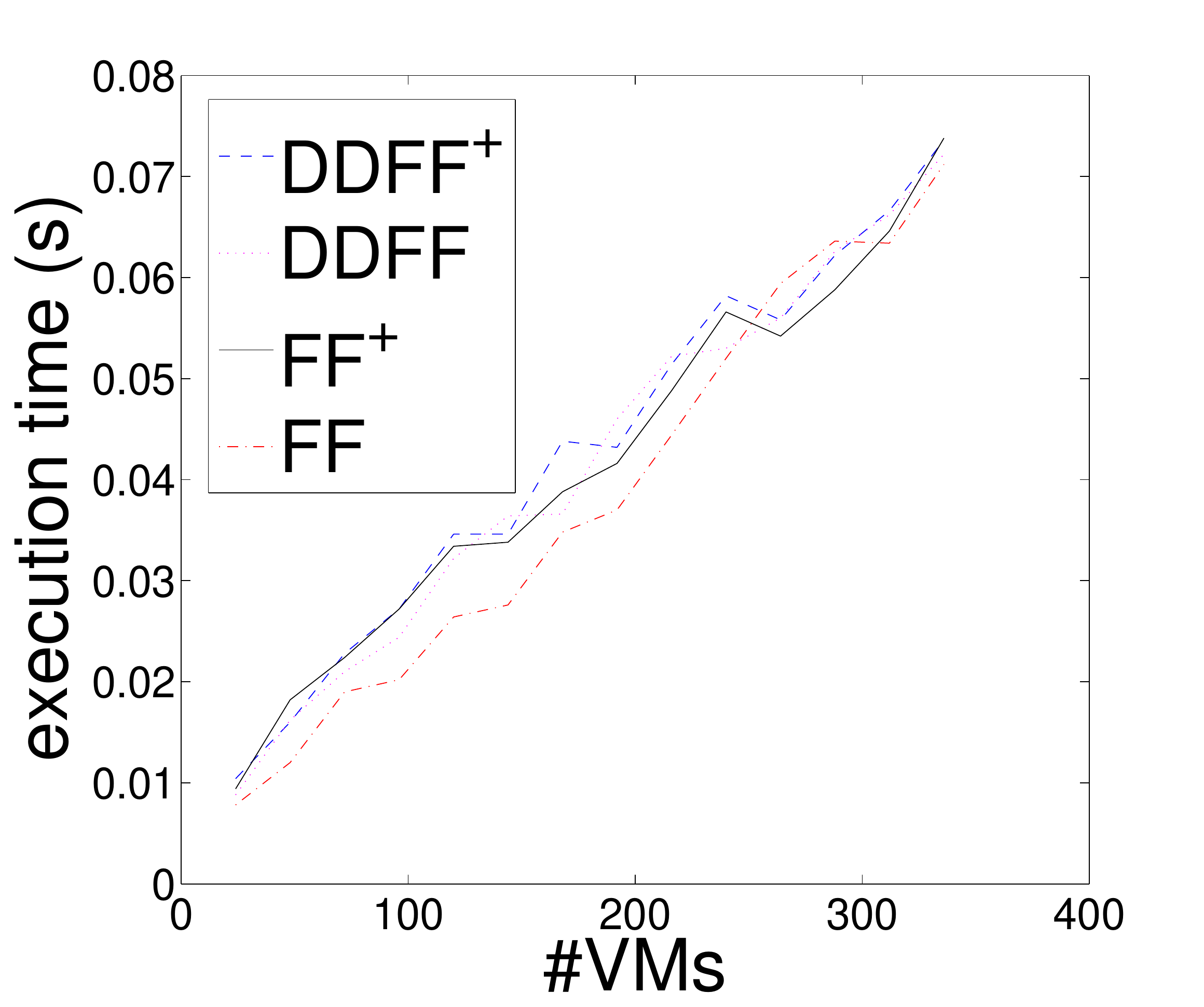}
			\caption{Effect of shuffling on execution time}
			\label{fig:shuffle:2}
		\end{subfigure}
		\caption{Effects of the shuffling process}
		\label{fig:shuffle}
	\end{figure}
	
	In \Cref{fig:shuffle:1}, the lines representing the \#servers required by FF$^{+}$ and DDFF$^{+}$ lie below those for FF and DDFF, indicating that the shuffling process reduces the total \#servers required. Moreover, the effectiveness of the shuffling process becomes more evident as the \#VMs increases. Furthermore, the two nearly overlapping lines in~\Cref{fig:shuffle:1} indicate that the duration-descending process does not have much impact on the \#servers under our proposed heterogeneous and multidimensional CDBP model.  \par
	Regarding the execution time, \Cref{fig:shuffle:2} implies that the shuffling process does not incur much extra time. Although FF and DDFF require shorter execution times when the \#VMs is small, the difference disappears as the \#VMs increases. The extra execution time incurred by the shuffling process is thus regarded as trivial compared to the time required for the total scheduling process and the perturbations caused by different server orders.\par
	
	From the experimental results in this subsection, it can be concluded that the shuffling process can slightly reduce the \#servers required by FF and DDFF. Moreover, the additional time incurred by the shuffling process is trivial, particularly when the \#VMs is large.
	
	\subsection{Influence of the Number of Virtual Machines}
	\label{subsec:number}
	
	In this subsection, the algorithms are evaluated and compared on workload \RNum{2}, in which the \#VMs varies. \Cref{table:number} lists the average \#servers and execution time of each algorithm for each \#VMs. Data for BB are not available when the \#VMs exceeds 240 because the computational resources of the experimental environment are no longer sufficient to support BB in these cases.
	
	\begin{table}
		\caption{Evaluations and comparisons of the algorithms with various \#VMs. \#S and T denote the \#servers and the execution time, respectively. The fewest \#servers and the shortest execution time achieved among all algorithms are marked with \underline{underscores} and \textbf{bold text}, respectively. }
		\label{table:number}       
		\scriptsize
		\centering
		
		\begin{tabular*}{\textwidth}{*{10}{c|}c}
			\hline
			\tiny
			
			\multirow{2}*{\#VMs} &\multicolumn{2}{c}{DCBB} & \multicolumn{2}{|c}{BB} & \multicolumn{2}{|c}{OEMACS$^+$}  & \multicolumn{2}{|c}{DDFF$^+$} & \multicolumn{2}{|c}{FF$^+$}   \\ 
			\cline{2-11}
			&	\#S	&	T (s) 	&	\#S	&	T (s)	&	\#S	&	T (s)	&	\#S	&	T (s)	&	\#S	&	T (s)	\\ \hline
			
			24      &       \underline{10.0}        &       1.2     &       \underline{10.0}        &       39.0    &       12.0    &       4.8        &    11.6    &       0.010   &       11.4    &       \textbf{0.009}  \\ \hline
			48      &       20.0                    &       45.8    &       \underline{19.8}        &       831.0   &       23.6    &       8.3        &    22.0    &       0.016   &       22.6    &       \textbf{0.018}  \\ \hline
			72      &       29.2                    &       33.2    &       \underline{29.0}        &       479.1   &       34.8    &       9.7        &    33.2    &       0.023   &       33.2    &       \textbf{0.022}  \\ \hline
			96      &       \underline{39.0}        &       50.4    &       \underline{39.0}        &       1049.5  &       46.4    &       15.5       &    44.4    &       \textbf{0.027}  &       44.4    &       \textbf{0.027}  \\ \hline
			120     &       49.0                    &       50.5    &       \underline{48.2}        &       1027.6  &       58.4    &       17.5       &    55.8    &       0.035   &       54.8    &       \textbf{0.033}  \\ \hline
			144     &       59.0                    &       50.6    &       \underline{58.0}        &       1034.4  &       69.2    &       22.9       &    67.2    &       0.035   &       65.8    &       \textbf{0.034}  \\ \hline
			168     &       68.0                    &       51.3    &       \underline{67.8}        &       1037.2  &       81.2    &       28.1       &    76.4    &       0.044   &       77.6    &       \textbf{0.039}  \\ \hline
			192     &       \underline{78.4}        &       52.7    &       \underline{78.4}        &       1044.6  &       92.6    &       37.6       &    88.8    &       0.043   &       87.8    &       \textbf{0.042}  \\ \hline
			216     &       \underline{88.0}        &       51.9    &       89.0                    &       1053.8  &       104.8   &       48.4       &    99.4    &       0.051   &       99.6    &       \textbf{0.049}  \\ \hline
			240     &       \underline{97.8}        &       50.9    &       98.0                    &       1051.8  &       115.6   &       51.5       &    111.0   &       0.058   &       110.6   &       \textbf{0.057}  \\ \hline
			264     &       \underline{110.0}       &       52.3    &       NaN                     &       NaN     &       126.6   &       68.4       &    121.0   &       0.056   &       119.6   &       \textbf{0.054}  \\ \hline
			288     &       \underline{118.8}       &       52.4    &       NaN                     &       NaN     &       138.8   &       76.2       &    132.2   &       0.062   &       131.6   &       \textbf{0.059}  \\ \hline
			312     &       \underline{127.6}       &       51.9    &       NaN                     &       NaN     &       150.6   &       89.4       &    142.0   &       0.067   &       142.2   &       \textbf{0.065}  \\ \hline
			336     &       \underline{139.8}       &       56.4    &       NaN                     &       NaN     &       162.4   &       88.8       &    153.8   &       \textbf{0.074}  &       152.8   &       \textbf{0.074}  \\ \hline
			
		\end{tabular*}
	\end{table}

	\Cref{table:number} shows that DCBB yields better results than DDFF$^{+}$, FF$^{+}$, and OEMACS$^{+}$ do in much less time than that required by BB. Although BB generally yields a slightly smaller \#servers than DCBB does when the \#VMs is less than 192, it requires a longer execution time and an enormous amount of computational resources. Furthermore, DCBB yields the smallest \#servers when the \#VMs exceeds 192. The performance of OEMACS$^{+}$ is not desirable, as it often requires both the largest \#servers and a relatively long execution time. In addition, DDFF$^{+}$ and FF$^{+}$ always output similar results for the \#VMs in less than 0.1 s, which is a trivial execution time compared to those of the other algorithms. \par
	The results in this subsection show that a larger \#VMs can cause an increase in the resulting \#servers. Moreover, BB yields the smallest \#VMs, though with a long execution time, when the problem scale is small. However, ideally accurate algorithms (e.g., BB) cannot handle large-scale problems because of their prohibitive computational complexity. In addition, this experiment demonstrates that DCBB can achieve a suitable tradeoff between accuracy and efficiency, as it can output near-optimal results within 60 s. Although DDFF$^{+}$ and FF$^{+}$ do not yield better results than that of DCBB, they can produce a scheduling solution within less than 0.08 s, which is suitable for real-time scheduling.
	
	\subsection{Influence of Time Factors}
	\label{subsec:time}
	\begin{figure}
		\begin{subfigure}{0.45\textwidth}
			
			\includegraphics[width=\linewidth]{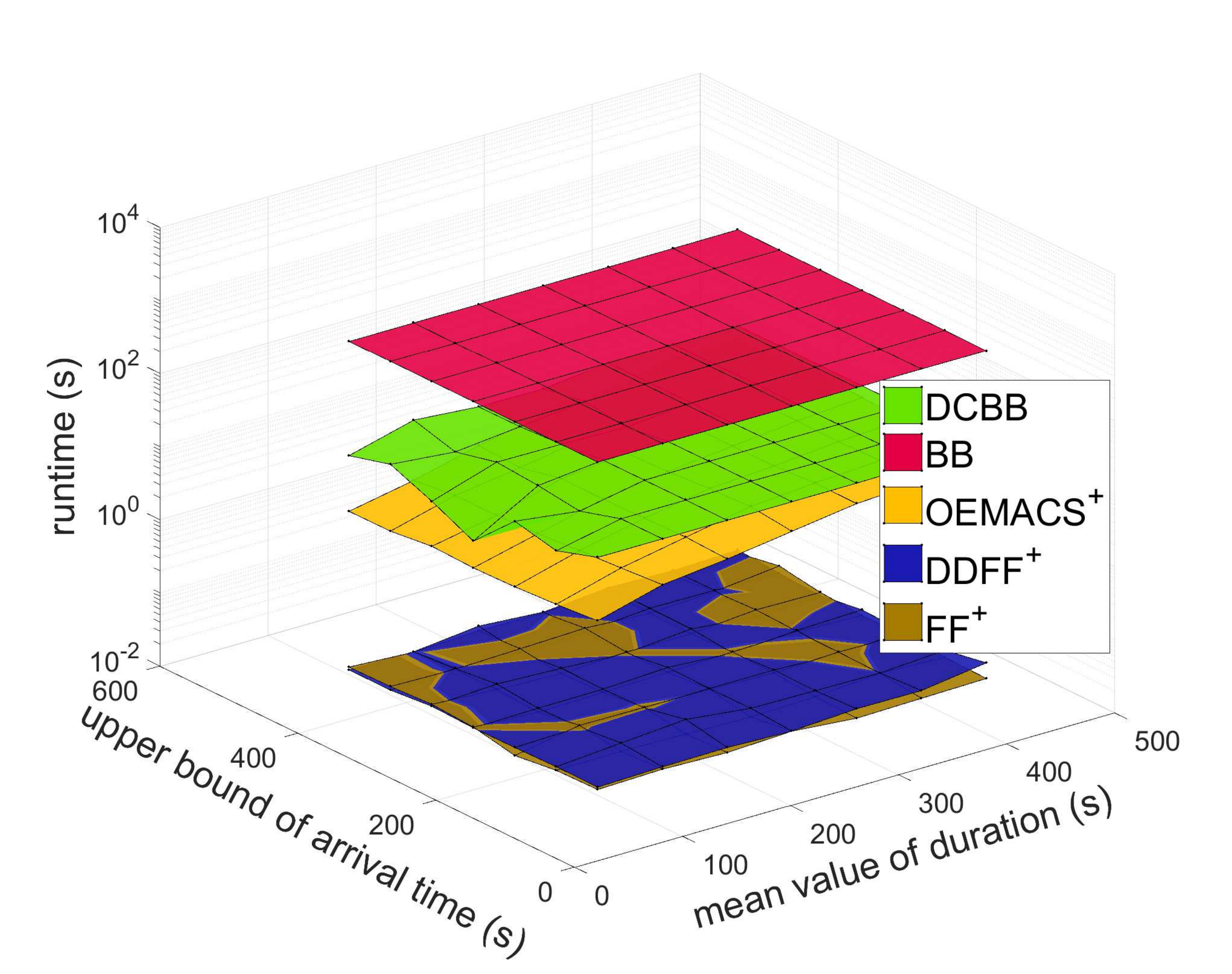}
			\caption{Execution times required by the algorithms under different VM time configurations}
			\label{fig:time:1}
			
		\end{subfigure} \hfill
		\begin{subfigure}{0.45\textwidth}
			
			\includegraphics[width=\linewidth]{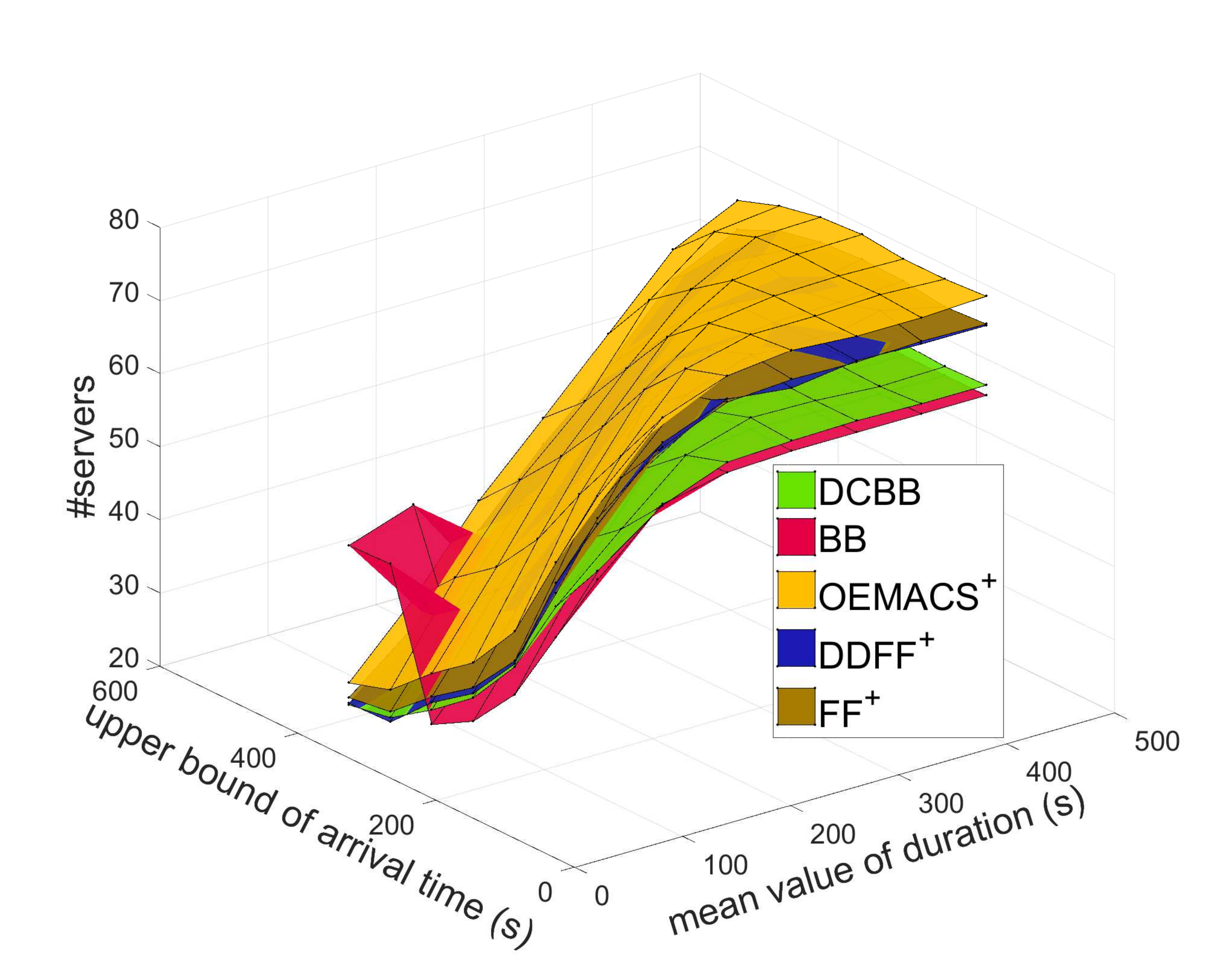}
			\caption{\#servers required by the algorithms under different VM time configurations}
			\label{fig:time:2}
		\end{subfigure}
		\bigskip
		\begin{subfigure}{0.45\textwidth}
			\centering
			\includegraphics[width=.9\linewidth]{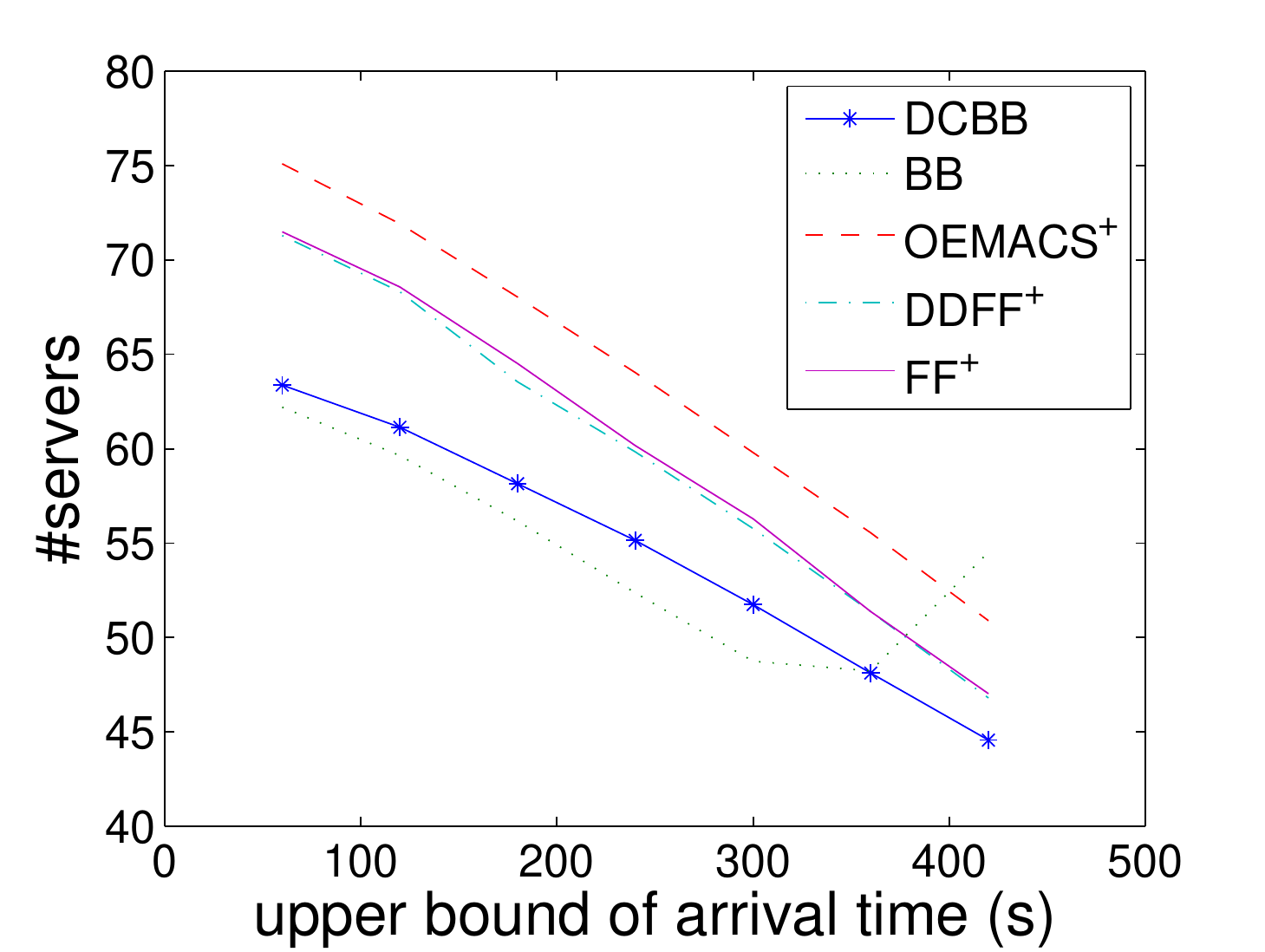}
			\caption{\#servers required versus the arrival time range}
			\label{fig:time:3}
		\end{subfigure} \hfill
		\begin{subfigure}{0.45\textwidth}
			\centering
			\includegraphics[width=.9\linewidth]{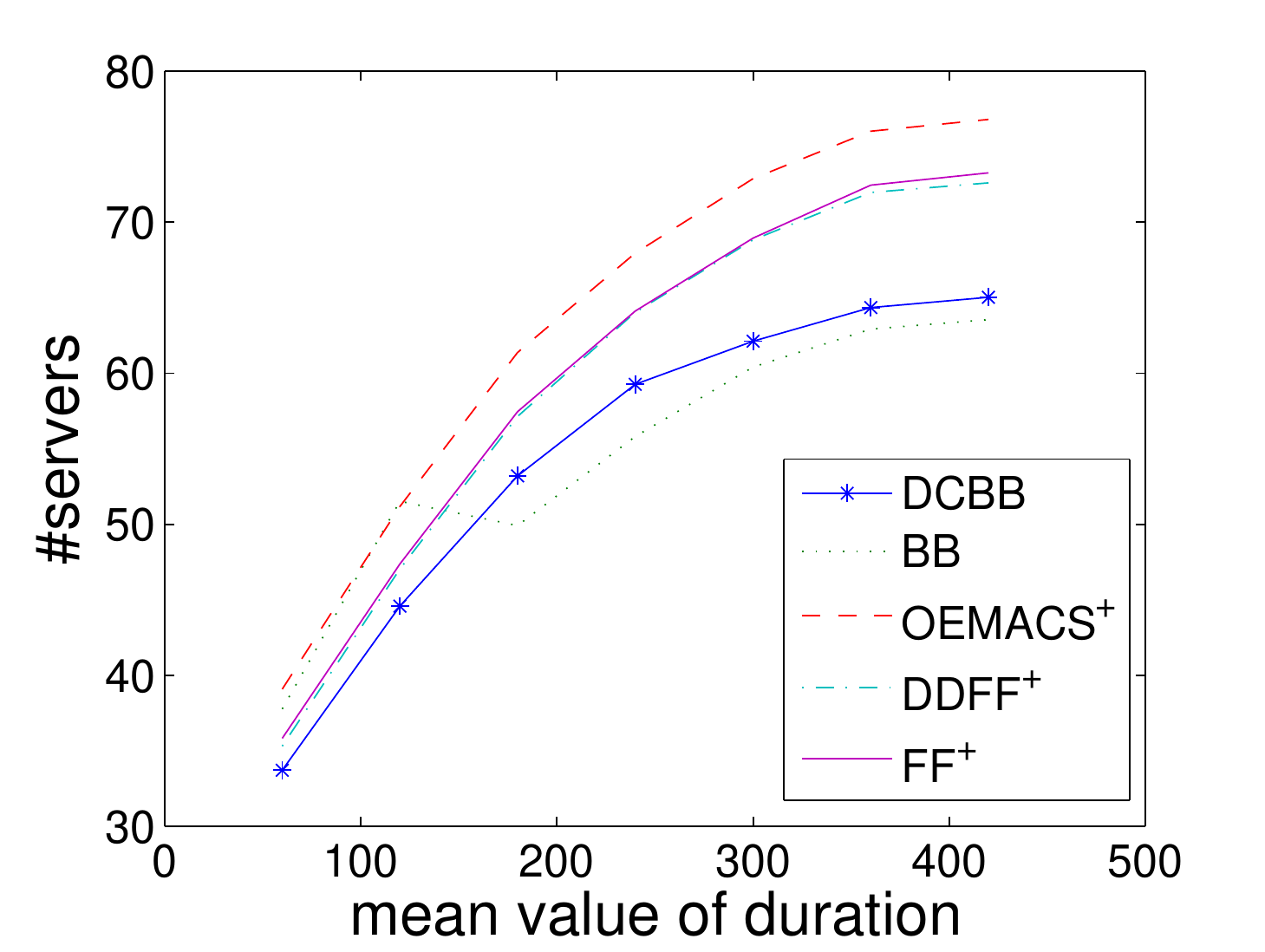}
			\caption{\#servers required versus the mean duration}
			\label{fig:time:4}
		\end{subfigure}
		\caption{Evaluations and comparisons of the algorithms with various VM time distributions}
		\label{fig:time}
	\end{figure}
	To investigate the influence of the VM arrival times and durations, the results obtained by the algorithms on Workload \RNum{3}, with varying distributions of the arrival times and durations of the VMs, are presented to show the resulting changes in the accuracy and efficiency of the different algorithms. \Cref{fig:time} compares the results obtained with various upper bounds on the arrival time and various mean durations. \par
	As shown in \Cref{fig:time:1}, the execution times of the algorithms do not change much with the variation of the time factors. Consistent with the previous results, BB takes the most time (approximately 1020 s), while DDFF$^{+}$ and FF$^{+}$ require only a trivial execution time (less than 0.1 s). Because BB does not terminate before the time limit (1000 s) is reached, it does not guarantee an optimal result. \par
	From a comprehensive analysis of \Cref{fig:time:2,fig:time:3,fig:time:4}, it can be concluded that a shorter duration and a wider arrival time range both result in a lower \#servers required by algorithms. These results are logical, as a sparser distribution of VMs can lead to a higher degree of time multiplexing, thus contributing to a smaller \#servers. Although BB generally requires the fewest \#servers, its execution time is long. Moreover, several outliers are produced by BB, as seen in \Cref{fig:time:2}, reflecting its instability under a time limit. Similar to the previous results, OEMACS$^{+}$ generally performs the worst in this experiment. Furthermore, DCBB yields the second smallest \#servers, as shown in \Cref{fig:time:2,fig:time:3,fig:time:4}. \par
	
	To summarize, a shorter duration and a wider arrival time range cause the algorithms to require a lower \#servers, possibly because of a higher degree of time multiplexing. The execution time does not vary much with different time parameter settings. Consistent with the previous results, DCBB yields a near-optimal solution within a relatively short execution time, BB achieves the lowest \#servers with the longest execution time, OEMACS$^{+}$ typically delivers unsatisfactory performance, and the FF-based algorithms have the fastest processing speed.
	
	\subsection{Summary}
	\label{subsec:summary}
	The experiments presented in \Cref{subsec:real world,subsec:convergence speed,subsec:shuffle,subsec:number,subsec:time} have answered questions Q1-Q5 posed at the beginning of \Cref{sec:implementation and experiments}:
	\begin{enumerate}[start=1,label={\upshape\bfseries A\arabic*.},wide = 0pt, leftmargin = 2.2em]
		\item BB and DCBB can both yield the optimal solutions on the real-world datasets considered here; however, DCBB requires an order of magnitude less time than BB does. Meanwhile, FF$^{+}$ and DDFF$^{+}$ can produce a scheduling solution within an insignificant execution time (less than 0.1 s), indicating that they are suitable for real-time scheduling. However, the performance of OEMACS$^{+}$ is not satisfactory in terms of either accuracy or efficiency. 
		\item DCBB has a much faster convergence rate than that of BB. As for OEMACS$^{+}$, its scheduling solutions show almost no improvement with increasing execution time.
		\item The shuffling process improves the accuracy of DDFF and FF with a trivial increase in the execution time.
		\item A larger \#VMs results in an increase in the \#servers required. Furthermore, ideally accurate algorithms such as BB have difficulty handling large-scale problems because of their prohibitive computational complexity.
		\item A shorter mean duration and a wider arrival time range for the VMs can result in a lower $\#$servers while exerting little influence on the execution time. 
	\end{enumerate}
	
	In addition, the experiments also demonstrate that the proposed algorithms satisfy requirements \textbf{R1-4} mentioned in \Cref{subsed: motivations and requirements}. Since the algorithms can handle Workloads I-III, for which the resources are multidimensional and the servers are heterogeneous, \textbf{R1} and \textbf{R2} are satisfied. Furthermore, \textbf{R3} is met because different VMs without overlapping execution times can share the same resources. \textbf{R4} is also fulfilled since no VM requests are rejected in the experiments. 
	
	Overall, the experimental data confirm that DCBB can yield near-optimal scheduling solutions while having faster convergence rate than the other evaluated search-based algorithms do. The results also demonstrate that the FF-based algorithms have the fastest processing speed and that BB can produce the best solution when the problem scale is small. In addition, the experimental results for OEMACS$^{+}$ are unsatisfactory, possibly because of the extra problem complexity introduced by the additional time and resource dimensions. Furthermore, a wider arrival time range and a shorter mean duration for the VMs both cause a lower \#servers to be required since they enable higher degrees of time multiplexing. \par

	\section{Conclusions and Future Work} 
	\label{sec:conclusion and future works}
	To lower the expensive upfront cost of private clouds, this paper proposes DCBB, an effective and efficient VMP algorithm that is applicable to the heterogeneous and multidimensional CDBP model, to reduce the \#servers required to accommodate VMs. The proposed model and algorithm employ time multiplexing to achieve more efficient and flexible scheduling. Theoretical analyses have been conducted to identify the upper bound and other features of DCBB. The experimental data clearly confirm the superiority of DCBB. It has been verified that DCBB can achieve near-optimal solutions while requiring significantly less execution time (by an order of magnitude on a real-world workload) than the BB algorithm does. The experimental results also show that DCBB has a much faster convergence rate than those of the other search-based algorithms evaluated. Although the BB algorithm can yield the optimal solution in theory, it requires a long execution time and a large amount of computational resources and shows unstable performance when given a time limit. Moreover, the accuracies of the DDFF and FF algorithms have been improved by including an additional shuffling process, and the resulting algorithms can be applied for real-time scheduling because of their trivial processing time. In addition, the experimental results demonstrate that OEMACS$^{+}$ does not deliver the expected performance under the proposed model, possibly because of the extra problem complexity introduced by the additional time and resource dimensions. Furthermore, the experiments indicate that, in addition to a lower \#VMs, a shorter mean duration and a wider arrival time range for the VMs can also cause a lower \#servers to be required due to the higher degree of time multiplexing that can be achieved in this case.\par
	Although extensive experiments have been conducted to evaluate and compare the algorithms considered here, the superiority of DCBB has not been fully theoretically proven. In addition, the influence of the adopted clustering algorithm on DCBB is not clear. Therefore, further theoretical analysis should be conducted to discover more features of DCBB and to enable further improvements.
	
	

	
	\clearpage


\end{document}